\documentclass{elsarticle}
\usepackage{amsmath}
\usepackage{amssymb}
\usepackage{booktabs}
\usepackage{stmaryrd}
\usepackage{url}
\usepackage{todonotes}
\usepackage{pdfpages}
\usepackage{float}
\usepackage{amsfonts, tikz}
\usepackage{subfigure}
\usetikzlibrary{patterns}
\usepackage[makeroom]{cancel}
\usetikzlibrary{shapes,shapes.geometric,arrows,fit,calc,positioning,automata,}
\usepackage{footnote}
\usepackage[bottom]{footmisc}

\usepackage{array}
\usepackage{amsthm}
\usepackage{fixltx2e}

\usepackage{todonotes}

\DeclareMathOperator*{\argmin}{arg\,min}

\newtheorem{theorem}{Theorem}
\newtheorem{proposition}{Proposition}
\newtheorem{lemma}{Lemma}

\begin{document}

\begin{frontmatter}

\title{Boolean kernels for collaborative filtering in top-N item recommendation}

\author{Mirko Polato\footnote{Corresponding author. E-mail address: mpolato@math.unipd.it}}
\address{University of Padova - Department of Mathematics\\ Via Trieste, 63, 35121 Padova - Italy}

\author{Fabio Aiolli}
\address{University of Padova - Department of Mathematics\\ Via Trieste, 63, 35121 Padova - Italy}

\begin{abstract}
	In many personalized recommendation problems available data consists only of positive interactions (implicit feedback) between users and items. This problem is also known as One-Class Collaborative Filtering (OC-CF).
	Linear models usually achieve state-of-the-art performances on OC-CF problems and many efforts have been devoted to build more expressive and complex representations able to improve the recommendations.
	Recent analysis show that collaborative filtering (CF) datasets have peculiar characteristics such as high sparsity and a long tailed distribution of the ratings.
	In this paper we propose a boolean kernel, called Disjunctive kernel, which is less expressive than the linear one but it is able to alleviate the sparsity issue in CF contexts.
	The embedding of this kernel is composed by all the combinations of a certain arity $d$ of the input variables, and these combined features are semantically interpreted as disjunctions of the input variables.
	Experiments on several CF datasets show the effectiveness and the efficiency of the proposed kernel.
\end{abstract}


\end{frontmatter}


\section{Introduction}
Collaborative Filtering (CF) is the \textit{de facto} approach for making personalized recommendation. CF techniques exploit  historical information about the user-item interactions in order to improve future recommendations to users.
User-item interactions can be of two types: explicit or implicit. Explicit feedback is an unambiguous information about how much a user likes/dislikes an item and it is usually represented by a rating (e.g., 1 to 5 stars scale, thumbs up vs. thumbs down). Conversely, implicit feedback is a binary information because it simply means the presence or the absence of an interaction between a user and an item, and it is by its nature ambiguous. 

Even though, in the past, the explicit setting have got most of the attention by the research community, recently the focus is drifting towards the implicit feedback context (also known as One-Class CF problem, OC-CF) because of the following two main reasons: $(i)$ implicit data are much more easy to gather as they do not require any active action by the user, and $(ii)$ they are simply more common.

Unlike the explicit feedback setting where the recommendation task is the rating prediction, in the implicit feedback case the goal is to produce an ordered list of items, where those items that are the most likely to have a future interaction with the user are at the top. In this work we focus on this last type of task which is known as top-N recommendation.

The first developed approaches for OC-CF problems are the neighbourhood-based ones \cite{Desrosiers:2011}. Despite they do not employ any kind of learning, they have been shown to be very effective \cite{Aiolli:2013,Polato-recsys:2016}.
Afterwards, methods employing learning have been proposed, such as \textit{SLIM} \cite{Karypis:2011}, \textit{WRMF} \cite{Koren:2008}, \textit{CF-OMD} \cite{Aiolli:2014,Polato:2016}, and more recently, \textit{LRec} \cite{Sedhain:2016}, and \textit{GLSLIM} \cite{Christakopoulou:2016}.
Both the neighbourhood-based and the learning-based methods mentioned above exploit linear relations between users and/or items. 
The effectiveness of linear models in CF is further underlined in \cite{Bresler:2014} where the authors propose an online linear model and they demonstrate its performance guarantees. 

Recently, see \cite{Polato:2016} and \cite{Polato:2017}, an efficient kernel-based method for OC-CF has been proposed. This method, called \textit{CF-KOMD}, is based on the approach presented in \cite{Aiolli:2014} which showed  better performance than other well known recommendation algorithms such as \textit{SLIM}, \textit{WRMF} and \textit{BPR}. \textit{CF-KOMD} have been tested using different kernels such as linear, polynomial, and Tanimoto kernel, and it has shown state-of-the-art performance on many CF datasets. From the reported results it is possible to notice that the linear kernel achieves very good results despite being the least expressive.

This behaviour is typical in CF datasets because they usually own the following two characteristics \cite{Polato:2017,Zhou:2010,Grcar:2006}: they are very sparse ($\sim 1\%$ density), and the distribution of the interactions over the items and/or over the users are long tailed. For these reasons, for this kind of data, it is not ideal to use more complex and, consequently, even more sparse representations. 

In \cite{Donini:2016}, Donini et al. characterized the notion of expressiveness of a representation (kernel function): more general representations correspond to kernels constructed on simpler features (e.g. single variables, linear kernel), while more specific representations correspond to kernels defined on elaborated features (e.g. product of variables, polynomials). 
Intuitively, less expressive representations tend to emphasize more the similarities between the examples, with the extreme case being the constant kernel matrix in which each example is equal to the others. On the contrary, more expressive kernels highlight more the differences between the examples, with the extreme case being the identity (kernel) matrix in which examples are orthogonal to each other.
Different data can require different levels of expressiveness and, for the considerations expressed above, it is reasonable that CF contexts tend to favor more general representations (e.g., linear).\\

In this paper, we propose a new representation for boolean valued data which is less expressive than the linear one. Specifically, we propose a (boolean) kernel, called Disjunctive kernel (D-Kernel), in which the feature space is composed by all the combinations of the input variables of a given arity $d$, and these combined features are semantically interpreted as logical disjunctions of the input variables.

The underpinning idea behind the proposed D-Kernel is to define higher-level features that are a kind of generalization of the linear ones, so to obtain more general representations that hopefully can alleviate the sparsity issue. It is easy to see that the associated feature space cannot be explicitly defined because of the combinatorial explosion of the number of dimensions. For this, we give an efficient and elegant way to compute the kernel values which does not require an explicit mapping of examples over the feature space.

Besides the kernel definition, we also discuss about its properties and we demonstrate, both theoretically and empirically, that the expressiveness of the D-Kernel is monotonically decreasing with the arity $d$.

For the sake of completeness, and to confirm the ineffectiveness of more complex representations in OC-CF contexts, we also take into consideration an existing kernel that here we call Conjunctive kernel (C-Kernel). The embedding of this kernel is the same as the D-Kernel, that is all possible combinations of $d$ variables but, as the name suggests, the combined input variables are semantically interpreted as logical conjunctions. As opposed to the disjunctive case, we demonstrate that the expressiveness of the C-Kernel is monotonically increasing with the arity $d$.

Finally, we empirically assess the effectiveness and the efficiency of the proposed disjunctive kernel against other boolean kernels over six large CF datasets. Results show that, when using CF-KOMD \cite{Polato:2016} as kernel-based ranker, in almost all the datasets our kernel achieves the best AUC performances.\\

The rest of the paper is organized as follows. In Section \ref{sec:back} we introduce the notation used throughout the paper and the background knowledge needed to fully understand it. Section \ref{sec:bool} presents the existing boolean kernels and how they  relate to our work. Our main contributions are reported in Section \ref{sec:contr}, in particular the Disjunctive kernel is presented in Section \ref{sec:dkernel}. Finally, in Section \ref{sec:exp}, an extensive empirical work involving six different CF datasets is presented and, in Section \ref{sec:concl}, final considerations and future lines of research we can follow are proposed.

\section{Notation and Background} \label{sec:back}
In this section, we present the notation and the background notions useful to fully understand the remainder of this work.\\

Throughout the paper we generally consider learning problems with training examples $\{(\mathbf{x}_1,y_1),...,(\mathbf{x}_m,y_m)\}$ where $x_i \in \mathbb{B}^n$, with $ \mathbb{B} \equiv \{0,1\}$, and $y_i \in \{-1,+1\}$. $\mathbf{X} \in \mathbb{B}^{m\times n}$ denotes the binary matrix where examples are arranged in rows with the corresponding labels vector $\mathbf{y} \in \{-1,1\}^m$.   

A generic entry of a matrix $\mathbf{M}$ is indicated by $\mathbf{M}_{ij}$. 
$\mathbf{1}_n$ represents the column vector of dimension $n$ with all entries equal to 1. When not differently indicated, the norm $\|\cdot\|$ refers to the 2-norm of vectors, while $\|\cdot\|_{F}$ and $\|\cdot\|_{T}$ refer to the Frobenius norm and the trace  (a.k.a. nuclear) matrix norm, respectively.

Since in our experiments we focus on the top-N recommendation problem, we also need to introduce the basic elements of these kind of tasks. We call $\mathcal{U}$ the set of users, such that $|\mathcal{U}|=n$ (i.e., variables), $\mathcal{I}$ the set of items, such that $|\mathcal{I}|=m$ (i.e., examples) and $\mathcal{R} = \{(u,i)\}$ the set of ratings (i.e., labels).

We refer to the binary rating matrix with $\mathbf{R} = \{r_{ui}\} \in \mathbb{B}^{n \times m}$ with $r_{ui} = 1$ if $(u,i)\in \mathcal{R}$, where users are on the rows and items on the columns.
We add a subscription to both user and item sets to indicate, respectively, the set of items rated by a user $u$ ($\mathcal{I}_u$) and the set of users who rated the item $i$ ($\mathcal{U}_i$).
It is worth to notice that users are variables as long as we face the recommendation task using an item-based approach, otherwise, with a user-based method, items and users' roles should be inverted.

In the reminder, when not specified differently, $\mathbf{x}$ and $\mathbf{z}$ are considered boolean vectors of dimension $n$, i.e., $\mathbf{x}, \mathbf{z} \in \mathbb{B}^n$, and $\langle \mathbf{x}, \mathbf{z} \rangle$ stands for their dot product.

\subsection{CF-KOMD}\label{sec:cfkomd}
CF-KOMD is a recent kernel-based algorithm \cite{Polato:2016} for top-N recommendation inspired by preference learning \cite{Aiolli:2005}\cite{Aiolli:2008}, and designed to explicitly maximize the AUC (Area Under the ROC Curve).

This method, even though it is based on kernels, is very efficient, highly scalable and it is very suitable for datasets with few positive and many negative/unlabeled examples (e.g., CF datasets). 
In the following, we give a brief explanation of the algorithm, for further details please refer to \cite{Polato:2016}.

Let $\mathbf{W} \in \mathbb{R}^{n \times k}$ be the embeddings of users ($\mathcal{U}$) in a latent factor space and  $\mathbf{X} \in \mathbb{R}^{k \times m}$ be the embeddings of items ($\mathcal{I}$) in that space.
Given a user $u$, a ranking over items can be induced by the factorization $\hat{\mathbf{R}} = \mathbf{W}\mathbf{X}$, where $\hat{r}_{ui} = \mathbf{w}_u^\top \mathbf{x}_i$ with the constraint $\|\mathbf{w}_u\| = \|\mathbf{x}_i\| = 1$. 

CF-KOMD tries to learn (implicitly) the user representation $\mathbf{w}_u^*$, by solving the optimization problem

\begin{equation}\label{opt3}
	\boldsymbol{\alpha}_{u^+}^* = \argmin\limits_{\boldsymbol{\alpha}_{u^+}} \;\;\boldsymbol{\alpha}_{u^+}^\top \mathbf{K}_{u^+} \boldsymbol{\alpha}_{u^+} + \lambda_p\|\boldsymbol{\alpha}_{u^+}\|^2 - 2\boldsymbol{\alpha}_{u^+}^\top\mathbf{q}_u,
\end{equation}

where $\mathbf{\boldsymbol{\alpha}}_{u^+}$ is a probability distribution over the positive examples (i.e., rated items) for a given user $u$, $\mathbf{K} \in \mathbb{R}^{m \times m}$ is a kernel matrix between items induced by a given kernel function $\kappa$, and the elements of the vector $\mathbf{q}_u \in \mathbb{R}^{|\mathcal{I}_u|}$ are defined by 
$$
\quad q_{ui} = \frac{1}{|\mathcal{I}|} \sum_{j \in \mathcal{I}} \kappa(\mathbf{x}_i, \mathbf{x}_j).
$$ 

\noindent The induced ranking is obtained using the scoring function $$\hat{\mathbf{r}}_{u} = \mathbf{X}^{\top} \mathbf{w}_u^{*} =\mathbf{K}_{u^+:}^{\top} \boldsymbol{\alpha}_{u^+} - \mathbf{q}$$ where $\mathbf{K}_{u^+:} \in \mathbb{R}^{|\mathcal{I}_u|\times |\mathcal{I}|}$ is the matrix obtained by taking the subset of rows corresponding to the positive set of items for the users $u$, and $\mathbf{q} \in \mathbb{R}^{|\mathcal{I}|}$ is like $\mathbf{q}_u$ but defined over the whole set of items. Empirical results presented in \cite{Polato:2016}\cite{Polato:2017} have shown the effectiveness and the efficiency of this method.

\section{Boolean kernels}\label{sec:bool}


Boolean kernels are those kernels which interpret the input vectors as sets of boolean variables and  apply some boolean functions to them. 

If we restrict the input space to the set $\mathbb{B}^n$, then several kernels can be interpreted as boolean kernels since monomials (i.e., product of features) can be seen as conjunctions of positive boolean variables. In the following, we sketch some examples.

The first example is the linear kernel, i.e., $\kappa_{\textit{LIN}}(\mathbf{x}, \mathbf{z}) = \langle \mathbf{x}, \mathbf{z} \rangle$, in which the features correspond to the boolean literals. This kernel simply counts how many active boolean variables the input vectors have in common.

Another instance of a well known boolean kernel is the polynomial kernel \cite{Zhang:2003}\cite{Zhang:2005}, $\kappa_{\textit{POLY}}^d(\mathbf{x}, \mathbf{z}) = (\sigma \langle\mathbf{x},\mathbf{z}\rangle + c)^d, \sigma > 0, c \geq 0$, in which the feature space is represented by all the monomials up to the degree $d$. 
For example, given the degree $d=3$ and $\mathbf{x} \in \mathbb{B}^2$ the embedding would be composed by the features $x_1^3, x_1^2x_2, x_1x_2^2$ and $x_2^3$. It is worth to notice that, being $\mathbf{x} \in \mathbb{B}^2$, $x_1^2x_2$ and $x_1x_2^2$ are actually the same feature $x_1x_2$.

Similar to the polynomial kernel, but with a different combination of features, is the all-subset kernel, $\kappa_{\subseteq}(\mathbf{x},\mathbf{z}) = \prod_{i=1}^n (x_i z_i + 1)$, which considers a space with a feature for each subset of the input variables, including the empty subset. This generates all combination of features but, differently from the polynomial case, each factor in the monomial has degree 1. 

A restriction of the all-subset kernel is represented by the ANOVA kernel ($\kappa_A^d$) \cite{Shawe-Taylor:2004}, in which the embedding space is formed by monomials with a certain degree $d$. For example, given $\mathbf{x} \in \mathbb{B}^3$ the feature space for the all-subset kernel would be made by the features $x_1, x_2, x_3, x_1x_2, x_1x_3, x_2x_3, x_1x_2x_3$ and $\emptyset$, while for the ANOVA of degree 2 it would be formed by $x_1x_2, x_1x_3$ and $x_2x_3$.\\

All the above mentioned kernels are generic kernels that can have a boolean interpretation. However, there is a family of ``pure" boolean kernels thought for learning boolean functions. 
A well known boolean kernel is the Tanimoto Kernel \cite{Ralaivola:2005}, $\kappa_T(\mathbf{x},\mathbf{z}) = \frac{\langle \mathbf{x}, \mathbf{z} \rangle}{\|\mathbf{x}\|^{2} + \|\mathbf{z}\|^{2} - \langle \mathbf{x}, \mathbf{z} \rangle}$, which represents the Jaccard similarity coefficient in binary contexts.
Another important boolean kernel is the DNF (i.e., Disjunctive Normal Form) kernel \cite{Sadohara:2001}, $\kappa_{\textit{DNF}}(\mathbf{x}, \mathbf{z}) = -1+\prod_{i=1}^n (2x_i z_i-x_i-z_i+2)$, which induces a feature space spanned by all possible conjunctions. 
The main difference between the DNF kernel and all-subset kernel's feature space is that the DNF kernel considers also the variables in their negate version, as in the definition of DNF. So, for example, the feature $x_1 \wedge \bar{x_2} \wedge x_3$ is inside the feature space of $\kappa_{\textit{DNF}}$ but not inside the one of $\kappa_\subseteq$.
In \cite{Sadohara:2001} and \cite{KhardonRS:2005}, a variation of the DNF kernel is also considered in which only the positive conjunctions are taken into account and they call it Monotone DNF Kernel (mDNF Kernel), $\kappa_{\textit{mDNF}}(\mathbf{x}, \mathbf{z}) = -1 + \prod_{i=1}^n (x_i z_i + 1)$. We can note that $\kappa_{\textit{mDNF}}$ and $\kappa_\subseteq$ are the same kernel up to the constant -1.

In \cite{Zhang:2003}, Zhang et al. present a generalization of both the DNF and the mDNF kernel in which a parameter $\sigma \in \mathbb{R}_{+}$ is introduced which sets an inductive bias towards simpler ($0 < \sigma < 1$) or more complex ($\sigma > 1$) conjunctions. 

It is worth to notice that the DNF and the mDNF kernels are applicable to vectors with values in $\mathbb{R}^n$ even though they have been thought to learn boolean functions.
If we restrict the input to values in $\mathbb{B}^n$ their definitions can be simplified as follows:
$$
	\kappa_{\textit{DNF}}(\mathbf{x}, \mathbf{z}) = 2^{s(\mathbf{x},\mathbf{z})} - 1, \quad \kappa_{\textit{mDNF}}(\mathbf{x}, \mathbf{z}) = 2^{\langle \mathbf{x}, \mathbf{z} \rangle} - 1,
$$
where $s(\mathbf{x},\mathbf{z})$ counts how many ``bits" $\mathbf{x}$ and $\mathbf{z}$ have in common.

In order to have more control on the feature space, Sadohara et al. \cite{Sadohara:2002} propose a reduced variant of the DNF kernel in which only conjunctions with up to $d$ variables are considered. The $d$-DNF Kernel over binary vectors is defined as 
$$
\kappa_{\textit{DNF}}^d(\mathbf{x},\mathbf{z}) = \sum_{i=1}^d \binom{s(\mathbf{x},\mathbf{z})}{i}.
$$

In \cite{Sadohara:2002} this kernel is used inside the SVM framework for binary classification, in which the idea is to validate the value of $d$ in order to find the right trade-off between specialization and generalization.

In \cite{KowalczykSW:2001} and \cite{Khardon:2003}, boolean kernels are used for studying the learnability of logical formulas, specifically DNF, using maximum margin algorithms, such as SVM. In particular, in \cite{Khardon:2003} the authors show the learning limitations of these kind of kernels inside the PAC (Probably Approximately Correct) framework.

In \cite{Zhang:2005} Zhang et al. propose a Decision Rule Classifier based on boolean kernels which achieves, in many datasets, state-of-the-art results.\\

In this work we present a new boolean kernel, the Disjunctive kernel (D-Kernel), in which the monomials in the feature space are interpreted as disjunctions of boolean variables. Intuitively, the idea is to induce a kernel which is less expressive than the linear one. 

We also consider an existing boolean kernel, that we call Conjunctive kernel (C-Kernel), in which the embedding is the same as the D-Kernel, but the combined input variables are logically interpreted as conjunctions. For both these kernels we present their properties and the analysis of their expressiveness.

In the experimental section we compare these kernels against the linear kernel, the Tanimoto kernel and the mDNF Kernel. In particular, we do not test the DNF kernel because, in CF contexts with implicit feedback, there is no real negative information since the missing feedbacks are ambiguous.

\section{Structured boolean kernel families}\label{sec:contr}

In this section, we present two family of parametric boolean kernels in which the parameter plays a key role on the expressiveness of the kernel.

\subsection{Conjunctive kernel}\label{sec:ckernel}
In this section we consider a particular instantiation of the ANOVA kernel \cite{Shawe-Taylor:2004} where the input vectors are binary. In the ANOVA kernel 
of degree $d$, the feature space induced by the mapping function 
is represented by all monomials of degree $d$ that can be formed using combinations without repetition of the input variables.

As mentioned in Section \ref{sec:bool}, since we are focusing on binary input vectors, a monomial of degree $d$ can be interpreted as a conjunction of boolean variables, e.g., $x_1 x_3 x_7 \equiv x_1 \wedge x_3 \wedge x_7$, assuming 1 as \textit{true} and 0 as \textit{false}.

This boolean interpretation allows to see the ANOVA kernel of degree $d$ between $\mathbf{x}$ and $\mathbf{z}$ as the number of \textit{true} conjunctions of $d$ variables (i.e., literals) that $\mathbf{x}$ and $\mathbf{z}$ have in common.
In other words, the kernel counts how many conjunctions of $d$ input variables can be constructed using the set of active common features of $\mathbf{x}$ and $\mathbf{z}$. We call this kernel Conjunctive kernel (or simply C-Kernel for  brevity). 

Formally, the embedding of the C-Kernel \cite{Khardon:2003} of arity $d$ is given by
$$
	\boldsymbol{\phi}_\wedge^d : \mathbf{x} \mapsto (\boldsymbol{\phi}_\wedge^{(\mathbf{b})}(\mathbf{x}))_{\mathbf{b} \in \mathbb{B}_d},
$$
where $\mathbb{B}_d = \{\mathbf{b} \in \mathbb{B}^n \mid \|\mathbf{b}\|_1 = d\}$, and 
$$
\boldsymbol{\phi}_\wedge^{(\mathbf{b})}(\mathbf{x}) = \prod_{i=1}^n x_i^{b_i} = \mathbf{x}^\mathbf{b},
$$
where the notation $\mathbf{x}^\mathbf{b}$ stands for $x_1^{b_1} x_2^{b_2} \cdots x_n^{b_n}$.

The dimension of the resulting embedding is $\binom{n}{d}$, that is the number of possible combinations of $d$ different variables, while the resulting kernel $(\kappa_\wedge^d)$ is given by
$$
	\kappa_\wedge^d(\mathbf{x}, \mathbf{z}) = \langle \boldsymbol{\phi}_\wedge^d(\mathbf{x}), \boldsymbol{\phi}_\wedge^d(\mathbf{z}) \rangle = \sum_{\mathbf{b} \in \mathbb{B}_d} \mathbf{x}^{\mathbf{b}}\mathbf{z}^{\mathbf{b}} = \binom{\langle \mathbf{x},\mathbf{z}\rangle}{d}.
$$
This kernel is part of the results obtained in \cite{KowalczykSW:2001,KhardonRS:2005}. In both works, however, the proposed kernel is a linear combination of $\kappa_\wedge^j$ for $0 \leq j \leq d$.\\

The C-Kernel owns some peculiarities that are worth to mention:
\begin{itemize}
	\item the sparsity of the kernel increases with the increasing of the arity $d$. This is due to the fact that the number of active conjunctions decreases as the number of involved variables increases (see Figure \ref{fig:cgraph});
	\item the dimension of the feature space is not monotonically increasing with $d$, in fact, the dimension reaches its maximum when  $d = \lfloor \frac{n}{2} \rfloor$ and then starts to decrease;
	\item when $d=1$ the obtained kernel is equal to the linear one
$$
	\kappa_\wedge^1(\mathbf{x}, \mathbf{z}) = \binom{\langle \mathbf{x},\mathbf{z}\rangle}{1} = \langle \mathbf{x},\mathbf{z}\rangle = \kappa_{\textit{LIN}}(\mathbf{x}, \mathbf{z});
$$
	\item anytime an example $\mathbf{x}$ has a number of active variables smaller than the arity $d$, it will be mapped into the null vector in  feature space. From a computational stand point, this can cause issues, e.g. whenever the induced kernel has to be normalized;
	\item the computation of this kernel is very efficient.

	\begin{proposition}\label{prop:time_ck}
		Time complexity of $\kappa_\wedge^d$ computation is $\mathcal{O}(n+d)$.
	\end{proposition}
	\begin{proof}
		The complexity of the dot product $\langle \mathbf{x}, \mathbf{z} \rangle$ is linear in $n$, so $\mathcal{O}(n)$. For any $q \in \mathbb{N}$, the binomial coefficient $\binom{q}{d}$ can be calculated by $ \prod_{i=0}^{d-1} \frac{q-i}{i+1}$ and hence it has a complexity $\mathcal{O}(d)$. Therefore, the overall complexity is $\mathcal{O}(n+d)$.
	\end{proof}
	In practice, $d$ is order of magnitude smaller than $n$ (i.e., $d \ll n$) so the complexity of $k_\wedge^d$ is actually $\mathcal{O}(n)$.
\end{itemize}

\subsubsection{Connection with the Monotone DNF kernel}

As said in Section \ref{sec:bool}, both works \cite{KhardonRS:2005} and \cite{Sadohara:2001} independently proposed the Monotone DNF (Disjunctive Normal Form) kernel (mDNF Kernel), in which the feature space is represented by all conjunctions of positive boolean literals (i.e., variables). 

This kernel is applicable to points in $\mathbb{R}^n$ but if we restrict to the case in which the input vectors are in $\mathbb{B}^n$ then it can be computed as follows:
$$
	\kappa_{\textit{mDNF}}(\mathbf{x}, \mathbf{z}) = -1 + \prod_{i=1}^n (x_i z_i + 1) = 2^{\langle \mathbf{x}, \mathbf{z} \rangle} -1.
$$

Actually, the mDNF kernel counts how many (any degree) conjunctions of variables $\mathbf{x}$ and $\mathbf{z}$ have in common, and this is strictly related to what the C-Kernel does. In fact, we can express the mDNF Kernel as a linear combination of C-Kernels of arity $1 \leq d \leq n$ as in the following:
$$
\kappa_{\textit{mDNF}}(\mathbf{x}, \mathbf{z}) = 2^{\langle \mathbf{x}, \mathbf{z} \rangle} -1 = \sum_{d=0}^n \binom{\langle \mathbf{x}, \mathbf{z} \rangle}{d} - 1 = \sum_{d=1}^n \kappa_\wedge^d(\mathbf{x}, \mathbf{z}).
$$

\subsection{Disjunctive kernel}\label{sec:dkernel}

The embedding of the Disjunctive kernel (D-Kernel) is the same as the C-Kernel, but the logical interpretation is different. In the D-Kernel, as the name suggests, the combinations of variables inside the feature space represent disjunctions of variables, e.g., $x_1 x_2 x_5 \equiv x_1 \vee x_2 \vee x_5$.\\
 
Formally, the embedding of the D-Kernel of arity $d$ is given by
$$
	\boldsymbol{\phi}_\vee^d : \mathbf{x} \mapsto (\boldsymbol{\phi}_\vee^{(\mathbf{b})}(\mathbf{x}))_{\mathbf{b} \in \mathbb{B}_d},
$$
with 
$$
\boldsymbol{\phi}_\vee^{(\mathbf{b})}(\mathbf{x}) = H(\langle \mathbf{x}, \mathbf{b} \rangle) = H\left( \sum_{i=1}^n x_i{b_i} \right), 
$$
where $H : \mathbb{R} \rightarrow \mathbb{B}$ is the Heaviside step function defined as
$$
	H(n) = 
\begin{cases}
    1  & \quad \text{if } n > 0\\
    0  & \quad \text{if } n \leq 0\\
  \end{cases}.
$$

The dimension of the resulting embedding is $n\choose{d}$. 
From a logical point of view, the D-Kernel of arity $d$ between $\mathbf{x}$ and $\mathbf{z}$, is the number of \textit{true} disjuctions of $d$ variables that $\mathbf{x}$ and $\mathbf{z}$ have in common.
To be \textit{true}, a disjunction needs only one active literal, or conversely, to be \textit{false} all of them have to be inactive (i.e., \textit{false}). 

In order to ease the understanding of how a D-Kernel $(\kappa_\vee^d)$ is calculated, we rely on a set theory interpretation. With a slight abuse of notation let $\mathbf{x}$ and $\mathbf{z}$ be the sets containing the active variables of the respective vectors, and $\mathbf{f}$ be the set of all the variables. It will be clear from the context when $\mathbf{x}$ and $\mathbf{z}$ are treated as binary vectors or as sets.

Let denote $N_d(\mathbf{x}) = \binom{|\mathbf{x}|}{d}$, the number of combinations of size $d$ that can be formed using the elements of the set $\mathbf{x}$, hence $N_d(\mathbf{f}) = \binom{n}{d}$.

To get the value $\kappa_\vee^d(\mathbf{x}, \mathbf{z})$ we need to count all the $d$-disjunctions that contain at least an active variable of $\mathbf{x}$ and $\mathbf{z}$ (note that this variables can be the same). An easier way to get this counting is to face the problem in a negative fashion, that is, we subtract from all the possible $d$-combinations the ones with no active elements from both $\mathbf{x}$ and $\mathbf{z}$.

Formally:
\begin{align*}
	\kappa_\vee^d(\mathbf{x}, \mathbf{z}) &= \langle \boldsymbol{\phi}_\vee^d(\mathbf{x}), \boldsymbol{\phi}_\vee^d(\mathbf{z}) \rangle = N_d(\mathbf{f}) - N_d(\overline{\mathbf{x}}) - N_d(\overline{\mathbf{z}}) + N_d(\overline{\mathbf{x} \cup \mathbf{z}}) \\
	&= \binom{n}{d} - \binom{n - |\mathbf{x}|}{d} - \binom{n - |\mathbf{z}|}{d} + \binom{n - |\mathbf{x}| - |\mathbf{z}| + |\mathbf{x} \cap \mathbf{z}|}{d}.
\end{align*}

where the last term (i.e., $N_d(\overline{\mathbf{x} \cup \mathbf{z}})$) is needed since $- N_d(\overline{\mathbf{x}}) - N_d(\overline{\mathbf{z}})$ removes twice the combinations which contain active variables from neither $\mathbf{x}$ nor $\mathbf{z}$. 

Figure \ref{fig:set-dkernel} gives an illustration of how the kernel is calculated. Fig. \ref{fig:set-dkernel:a} represents all the possible combinations while Fig. \ref{fig:set-dkernel:b} and Fig. \ref{fig:set-dkernel:c} show the set of elements used to create the combinations that are discarded from the counting. Finally, Fig. \ref{fig:set-dkernel:d} illustrates the set of variables that are used to re-add the combinations that have been discarded twice.\\

As for the C-Kernel, when $d=1$ the D-Kernel is equal to the linear one
\begin{align*}
	\kappa^1_\vee(\mathbf{x}, \mathbf{z}) &= \binom{n}{1} - \binom{n - |\mathbf{x}|}{1} - \binom{n - |\mathbf{z}|}{1} + \binom{n - |\mathbf{x}| - |\mathbf{z}| + |\mathbf{x} \cap \mathbf{z}|}{1}\\
	&= n - (n - |\mathbf{x}|) - (n - |\mathbf{z}|) + (n - |\mathbf{x}| - |\mathbf{z}| + |\mathbf{x} \cap \mathbf{z}|)\\
	&= |\mathbf{x} \cap \mathbf{z}| := \langle \mathbf{x}, \mathbf{z} \rangle = \kappa_{\textit{LIN}}(\mathbf{x},\mathbf{z}).
\end{align*}

The D-Kernel has a special characteristic regarding its sparsity. In the following remark we prove that for any arity $d \geq 2$ the resulting kernel matrix is fully dense.\\

\noindent \textbf{Remark} For any dataset $\mathbf{X}$ and for any $d \geq 2$  the density of the D-Kernel matrix of degree $d$ induced by $\mathbf{X}$ (i.e., $\mathbf{K}_\mathbf{X}^{\vee_d}$) is $100\%$. It is easy to see that given two generic examples $\mathbf{x}$, $\mathbf{z}$ from the dataset $\mathbf{X}$ such that $\exists p,q \;|\; x_{p} = 1 \wedge z_{q} = 1$, then, for $d \geq 2$, the feature space induced by $\boldsymbol{\phi}_\vee^d$ has at least one feature which contains the \textit{or} between the \textit{p}-th and the \textit{q}-th variables of the vectors which is active in both  $\boldsymbol{\phi}_\vee^d(\mathbf{x})$ and $\boldsymbol{\phi}_\vee^d(\mathbf{z})$ and hence the kernel $\kappa_\vee^d(\mathbf{x},\mathbf{z}) \geq 1$. 

\begin{figure}[h!]
\centering     
\subfigure[$N_d(\mathbf{f})$]{\label{fig:set-dkernel:a}
\begin{tikzpicture}[pattern=north east lines, pattern color=blue]
	\scope
	\clip (-1.8,-1.5) rectangle (2.8,1.5);
	\fill (-1.8,-1.5) rectangle (2.8,1.5);
	\endscope
	\draw (0,0) ellipse (1.2 and 1) (-0.6,0.9)  node [text=black,above] {$\mathbf{x}$}
	      (1.1,0) ellipse (1 and 0.7) (1.6,0.6)  node [text=black,above] {$\mathbf{z}$}
	      (-1.8,-1.5) rectangle (2.8,1.5) node [text=black,right] {$\mathbf{f}$};
	\end{tikzpicture}
}
\subfigure[$-N_d(\mathbf{\overline{\mathbf{x}}})$]{\label{fig:set-dkernel:b}
\begin{tikzpicture}[pattern=north west lines, pattern color=red]
	\scope
	\clip (-1.8,-1.5) rectangle (2.8,1.5);
	\fill (-1.8,-1.5) rectangle (2.8,1.5);
	\endscope
	\scope
	\clip (0,0) ellipse (1.2 and 1);
	\fill[white] (0,0) ellipse (1.2 and 1);
	\endscope
	\draw (0,0) ellipse (1.2 and 1) (-0.6,0.9)  node [text=black,above] {$\mathbf{x}$}
	      (1.1,0) ellipse (1 and 0.7) (1.6,0.6)  node [text=black,above] {$\mathbf{z}$}
	      (-1.8,-1.5) rectangle (2.8,1.5) node [text=black,right] {$\mathbf{f}$};
	\end{tikzpicture}
}
\subfigure[$-N_d(\overline{\mathbf{z}})$]{\label{fig:set-dkernel:c}
\begin{tikzpicture}[pattern=north west lines, pattern color=red]
	\scope
	\clip (-1.8,-1.5) rectangle (2.8,1.5);
	\fill (-1.8,-1.5) rectangle (2.8,1.5);
	\endscope
	\scope
	\clip (1.1,0) ellipse (1 and 0.7);
	\fill[white] (1.1,0) ellipse (1 and 0.7);
	\endscope
	\draw (0,0) ellipse (1.2 and 1) (-0.6,0.9)  node [text=black,above] {$\mathbf{x}$}
	      (1.1,0) ellipse (1 and 0.7) (1.6,0.6)  node [text=black,above] {$\mathbf{z}$}
	      (-1.8,-1.5) rectangle (2.8,1.5) node [text=black,right] {$\mathbf{f}$};
	\end{tikzpicture}
}
\subfigure[$N_d(\overline{\mathbf{x}\cup\mathbf{z}})$]{\label{fig:set-dkernel:d}
\begin{tikzpicture}[pattern=north east lines, pattern color=blue]
	\scope
	\clip (-1.8,-1.5) rectangle (2.8,1.5);
	\fill (-1.8,-1.5) rectangle (2.8,1.5);
	\endscope
	\scope
	\clip (0,0) ellipse (1.2 and 1)
		  (1.1,0) ellipse (1 and 0.7);
	\fill[white] (0,0) ellipse (1.2 and 1)
				 (1.1,0) ellipse (1 and 0.7);
	\endscope
	\draw (0,0) ellipse (1.2 and 1) (-0.6,0.9)  node [text=black,above] {$\mathbf{x}$}
	      (1.1,0) ellipse (1 and 0.7) (1.6,0.6)  node [text=black,above] {$\mathbf{z}$}
	      (-1.8,-1.5) rectangle (2.8,1.5) node [text=black,right] {$\mathbf{f}$};
	\end{tikzpicture}
}
\caption{The blue striped sections, (a) and (d), show the set of elements used to create the \textit{active} combinations,  while the red ones, (b) and (c), are the combinations that are discarded from the counting.\label{fig:set-dkernel}}
\end{figure}
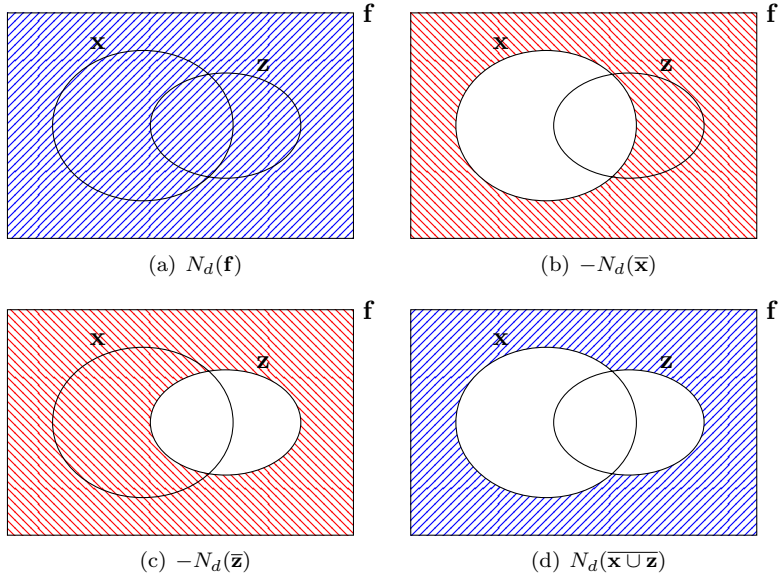

From a computational point of view, the D-Kernel has the same complexity as the C-Kernel.
\begin{proposition}
	Time complexity of $\kappa_\vee^d$ computation is $\mathcal{O}(n+d)$.
\end{proposition}
\begin{proof}
	As said in the proof of the Proposition \ref{prop:time_ck}, the complexity of $\binom{q}{d}$ computation is $O(d)$, while the complexity of $|\mathbf{x}|:=\|\mathbf{x}\|^2$, $|\mathbf{z}|:=\|\mathbf{z}\|^2$ and $|\mathbf{x} \cap \mathbf{z}| := \langle \mathbf{x}, \mathbf{z} \rangle$ is $\mathcal{O}(n)$.
	So, with the exception of the first binomial which has a complexity of $\mathcal{O}(d)$, all the binomials in $\kappa_\vee^d(\mathbf{x}, \mathbf{z})$ has a complexity of $\mathcal{O}(n+d)$ and hence the overall complexity is equals to $\mathcal{O}(n+d)$.
\end{proof}
Likewise the C-Kernel, in general, the arity $d$ is order of magnitude smaller than $n$ (i.e., $d \ll n$) and hence the complexity of $\kappa_\vee^d$ can be reduced to $\mathcal{O}(n)$.

To further speed up the computation of the D-Kernel matrix, we implemented a caching mechanism in order to avoid the recalculation of the same binomial coefficients. This trick led to a huge improvement in terms of computational time (see Section \ref{sec:time} for more details).

\subsection{Expressiveness of D-Kernel and C-Kernel}
In this section, we show some properties about the expressiveness of the above mentioned kernels. Firstly, we need to define the concept of expressiveness  and how it can be estimated.\\

The expressiveness of a kernel function can be roughly thought as the number of dichotomies that can be realized by a linear separator in that feature space. 
As discussed in \cite{Donini:2016}, the kernel's expressiveness can be captured by its rank since, it can be proved that,  let $\mathbf{K} \in \mathbb{R}^{m \times m}$ be a kernel matrix over a set of $m$ examples, and let $\textit{rank}(\mathbf{K})$ be its rank, then there exists at least one subset of examples of size $\textit{rank}(\mathbf{K})$ that can be shattered by a linear function.


In the same work the authors propose a novel measure, namely the \textit{spectral ratio}, which is easy to compute and it is demonstrated to be strictly related to the rank of the kernel matrix. 

In particular, the \textit{spectral ratio} is defined as: 
$$
	\mathcal{C}(\mathbf{K}) = \frac{\|\mathbf{K}\|_T}{\|\mathbf{K}\|_F} = \frac{\sum_i \mathbf{K}_{ii}}{\sqrt{\sum_{ij} \mathbf{K}_{ij}^2}}.
$$
We will use this measure to assess the complexity of the C-Kernel and the D-Kernel. Using the definition given in \cite{Donini:2016}, we say that a kernel function $\kappa_i$ is more general than another kernel function $\kappa_j$ when $\mathcal{C}(\mathbf{K}^{(i)}_{\mathbf{X}}) \leq \mathcal{C}(\mathbf{K}^{(j)}_\mathbf{X})$ for any possible dataset $\mathbf{X}$, and we use the notation  $\kappa_i \geq_G \kappa_j$.

In the experimental section we compare the kernels' expressiveness using the equivalent normalized version of the spectral ratio, that is:
$$
	\bar{\mathcal{C}} = \frac{\mathcal{C}(\mathbf{K})-1}{\sqrt{m}-1} \in [0,1].
$$

In the reminder of this section we will refer to the normalized version of a kernel function $\kappa$ with $\tilde{\kappa}$, which is defined as: 
$$
	\tilde{\kappa}(\mathbf{x},\mathbf{z})=\frac{\kappa(\mathbf{x},\mathbf{z})}{\sqrt{\kappa(\mathbf{x}, \mathbf{x}) \kappa(\mathbf{z}, \mathbf{z})}}.
$$

The use of normalized kernels leads to a nice property about the expressiveness of two kernel functions. Formally, let us define the following lemma.

\begin{lemma}
	\label{lemma}
	Let be given $\tilde{\kappa}$ and $\tilde{\kappa'}$, two normalized kernel functions, and the corresponding kernel matrices $\tilde{\mathbf{K}}_\mathbf{X}$, $\tilde{\mathbf{K}}'_\mathbf{X} \in \mathbb{R}^{m \times m}$ associated to a dataset $\mathbf{X}$ with $m$ examples. If, for any $\mathbf{x}_i, \mathbf{x}_j \in \mathbb{R}^m$, $\tilde{\kappa}(\mathbf{x}_i,\mathbf{x}_j)^2 \geq \tilde{\kappa}'(\mathbf{x}_i,\mathbf{x}_j)^2$ holds, then $\mathcal{C}(\tilde{\mathbf{K}}_\mathbf{X}) \leq \mathcal{C}(\tilde{\mathbf{K}}'_\mathbf{X})$ .
\end{lemma}

\begin{proof}\label{pf:lemma}
	Assume that $\tilde{\kappa}(\mathbf{x}_i,\mathbf{x}_j)^2 \geq \tilde{\kappa}'(\mathbf{x}_i,\mathbf{x}_j)^2$ holds for any $\mathbf{x}_i$, $\mathbf{x}_j$.
	
	Then, we can say that the following inequality is also true:
	\begin{equation}\label{ineq}
	\|\tilde{\mathbf{K}}_\mathbf{X}\|_F^2 = \sum_{i=1}^{m}\sum_{j=1}^{m} \tilde{\kappa}(\mathbf{x}_i,\mathbf{x}_j)^2 \geq \sum_{i=1}^{m}\sum_{j=1}^{m} \tilde{\kappa}'(\mathbf{x}_i,\mathbf{x}_j)^2 = \|\tilde{\mathbf{K}}'_\mathbf{X}\|_F^2.
	\end{equation}
	
	Since the kernels are normalized, $\|\tilde{\mathbf{K}}_\mathbf{X}\|_T=\|\tilde{\mathbf{K}}'_\mathbf{X}\|_T=m$ we have that:
	$$
	\mathcal{C}(\tilde{\mathbf{K}}_\mathbf{X}) = \frac{\|\tilde{\mathbf{K}}_\mathbf{X}\|_T}{\|\tilde{\mathbf{K}}_\mathbf{X}\|_F} = \frac{m}{\|\tilde{\mathbf{K}}_\mathbf{X}\|_F} \leq \frac{m}{\|\tilde{\mathbf{K}}'_\mathbf{X}\|_F} = \frac{\|\tilde{\mathbf{K}}'_\mathbf{X}\|_T}{\|\tilde{\mathbf{K}}'_\mathbf{X}\|_F} = \mathcal{C}(\tilde{\mathbf{K}}'_\mathbf{X}).
	$$
\end{proof}

\subsubsection{C-Kernel's expressiveness}\label{sec:exp-ck}
By construction, the features of a C-Kernel of arity $d$ have a clear dependence with the features of the same kernel of degree $d-1$.
In fact, consider the feature $x_1 x_2 x_3$ (i.e., $x_1\wedge x_2\wedge x_3$) of degree 3, its value is strictly connected to the ones of the features $x_1 x_2$, $x_1 x_3$ and $x_2 x_3$ of degree 2. 
In particular, the conjunctive feature $x_1 x_2 x_3$ will be active if and only if all the sub-features $x_1 x_2$, $x_1 x_3$ and $x_2 x_3$ are active, which, in turn, are active if and only if $x_1, x_2$ and $x_3$ are active too. Figure \ref{fig:cgraph} gives a visual idea of these dependencies. 

\begin{figure}[h]
\tikzset{elliptic state/.style={draw,ellipse}}
\tikzstyle{stuff_fill}=[preaction={fill=black!20}]
\centering
\begin{tikzpicture}[auto]
    \node[elliptic state, stuff_fill] (x1)                {$x_1$};
    \node[elliptic state, stuff_fill] (x2) [right =2cm of x1] {$x_2$};
    \node[elliptic state] (x3) [right =2cm of x2] {$x_3$};
    \node[elliptic state, stuff_fill] (x12) [below left =1cm and -.2cm of x1] {$x_1x_2:=x_1 \wedge x_2$};
    \node[elliptic state] (x13) [below =.85cm of x2] {$x_1x_3:=x_1 \wedge x_3$};
    \node[elliptic state] (x23) [below right=1cm and -.2cm of x3] {$x_2x_3:=x_2 \wedge x_3$};
    \node[elliptic state] (x123) [below =1cm of x13] {$x_1x_2x_3 := x_1 \wedge x_2 \wedge x_3$};
    \node (xml) [left =.1cm of x12]{};
    \node (xmr) [right =.1cm of x23]{};
    \node (t) [above =1.4cm of xml]{d=1};
	\path[->](x1) edge node [align=center]{$$} (x12) 
             (x1) edge node [align=center]{$$} (x13)
             (x2) edge node [align=center]{$$} (x12)
             (x2) edge node [align=center]{$$} (x23)
             (x3) edge node [align=center]{$$} (x13)
             (x3) edge node [align=center]{$$} (x23)
             (x12) edge node [align=center]{$$} (x123)
             (x13) edge node [align=center]{$$} (x123)
             (x23) edge node [align=center]{$$} (x123);
   \path (x123) -- coordinate (l32) (xml) -- coordinate (l21) (x1);
   \draw[dashed](xml |- l32) node[below] {d=3} -- (l32 -| xmr);
   \draw[dashed] (xml |- l21) node[below] {d=2}-- (l21 -| xmr);
\end{tikzpicture}
\caption{Dependencies between different arities of the C-Kernel. The nodes are features and the colored ones are active features. \label{fig:cgraph}}
\end{figure}
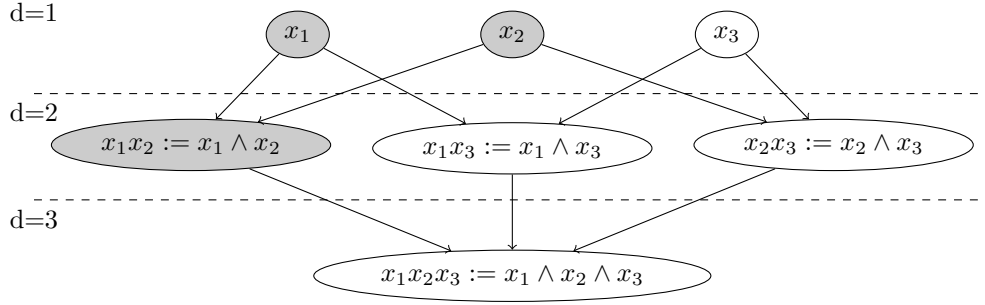

Intuitively, we expect that the higher the order of the C-Kernel, the higher the degree of sparsity of the corresponding kernel matrix.
We will prove this is true (for the normalized kernel) and we also show that the arity $d$ induces an order of expressiveness in the kernel functions $\kappa_\wedge^d(\mathbf{x}, \mathbf{z})$.

\begin{theorem}
	For any choice of $d \in \mathbb{N}$, the $d$-degree normalized C-Kernel $\tilde{\kappa}_\wedge^d$ is more general than $\tilde{\kappa}_\wedge^{d+1}$, that is $\tilde{\kappa}_\wedge^d \geq_G \tilde{\kappa}_\wedge^{d+1}$.
\end{theorem}

\begin{proof}\label{pf:conj}
	By Lemma \ref{lemma} it is sufficient to prove that 
	$$
	\tilde{\kappa}_\wedge^d(\mathbf{x}_i,\mathbf{x}_j)^2 \geq \tilde{\kappa}_\wedge^{d+1}(\mathbf{x}_i,\mathbf{x}_j)^2
	$$
	is true.
	Let us call $n_{i} = |\mathbf{x}_i|$, $n_{j} = |\mathbf{x}_j|$ and $n_{ij} = |\mathbf{x}_i \cap \mathbf{x}_j|$, then:
	\begin{align*}
		\tilde{\kappa}_\wedge^d(\mathbf{x}_i,\mathbf{x}_j)^2 &= 
		\frac{\kappa_\wedge^d(\mathbf{x}_i,\mathbf{x}_j)^2}{k_\wedge^d(\mathbf{x}_i,\mathbf{x}_i) \kappa_\wedge^d(\mathbf{x}_j,\mathbf{x}_j)} = 
		\frac{\binom{\langle\mathbf{x}_i, \mathbf{x}_j \rangle}{d}^2}{\binom{\langle\mathbf{x}_i, \mathbf{x}_i \rangle}{d} \binom{\langle\mathbf{x}_j, \mathbf{x}_j \rangle}{d}}\\ 
		&= \frac{N_d(\mathbf{x}_i\cap \mathbf{x}_j)^2}{N_d(\mathbf{x}_i)N_d(\mathbf{x}_j)} = \frac{\binom{n_{ij}}{d}}{\binom{n_i}{d}} \frac{\binom{n_{ij}}{d}}{\binom{n_j}{d}}.
	\end{align*}
	Now, let us examine the first term of the product, namely $\frac{\binom{n_{ij}}{d}}{\binom{n_i}{d}}$. We have:
	$$
		\frac{\binom{n_{ij}}{d}}{\binom{n_i}{d}} = \frac{\frac{n_{ij}!}{d!(n_{ij}-d)!}}{\frac{n_i!}{d!(n_i-d)!}} = \frac{n_{ij}!}{n_i!} \frac{(n_{i}-d)!}{(n_{ij}-d)!}.
	$$
	Since $n_{ij} \leq n_i$ by definition, we can say that:
	$$
		\frac{\binom{n_{ij}}{d}}{\binom{n_i}{d}} = \frac{n_{ij}!}{n_i!} \frac{(n_{i}-d)!}{(n_{ij}-d)!} = \frac{n_{ij}!}{n_i!} \frac{(n_i-d)(n_i-d-1)\dots(n_{ij}-d+1)(n_{ij}-d)!}{(n_{ij}-d)!},
	$$
	and so:
	\begin{align*}
		&\frac{\binom{n_{ij}}{d}}{\binom{n_i}{d}} = \frac{n_{ij}!}{n_i!} [(n_i-d)(n_i-d-1)\dots(n_{ij}-d+1)] \geq\\  &\frac{n_{ij}!}{n_i!} [(n_i-d-1)(n_i-d-2)\dots(n_{ij}-d)] =\frac{\binom{n_{ij}}{d+1}}{\binom{n_i}{d+1}}.
	\end{align*}
	The same inequality is clearly valid for the other term $\frac{\binom{n_{ij}}{d}}{\binom{n_j}{d}}$ and hence:
	$$
		\tilde{\kappa}_\wedge^d(\mathbf{x}_i,\mathbf{x}_j)^2 = \frac{\binom{n_{ij}}{d}}{\binom{n_i}{d}} \frac{\binom{n_{ij}}{d}}{\binom{n_j}{d}} \geq \frac{\binom{n_{ij}}{d+1}}{\binom{n_i}{d+1}} \frac{\binom{n_{ij}}{d+1}}{\binom{n_j}{d+1}} = \tilde{\kappa}_\wedge^{d+1}(\mathbf{x}_i,\mathbf{x}_j)^2.
	$$
\end{proof}

Note that the previous proof is valid anytime $d$ is smaller or equale than the examples with the minimum number of ones, otherwise the null vector would show up.
In Section \ref{sec:exp} it is empirically shown, for each dataset, how the normalized spectral ratio increases with the increasing of the degree.

\subsubsection{D-Kernel's expressiveness}\label{sec:exp-dk}
With similar considerations as for the C-Kernel, it is obvious that there exists a dependence between features of a D-Kernel of arity $d$ and $d+1$. As for the C-Kernel, let us consider the feature $x_1x_2x_3$ (i.e., $x_1\vee x_2 \vee x_3$) of arity 3. This will be active anytime at least one of the features $x_1x_2$, $x_2x_3$ or $x_1x_3$ of arity 2 is active.
Figure \ref{fig:dgraph} gives a visual idea of these dependencies. 

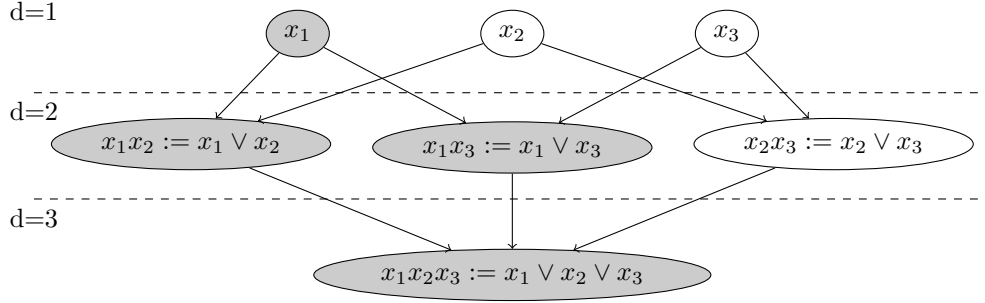
\begin{figure}[h]
\tikzset{elliptic state/.style={draw,ellipse}}
\tikzstyle{stuff_fill}=[preaction={fill=black!20}]
\centering
\begin{tikzpicture}[auto]
    \node[elliptic state, stuff_fill] (x1)                {$x_1$};
    \node[elliptic state] (x2) [right =2cm of x1] {$x_2$};
    \node[elliptic state] (x3) [right =2cm of x2] {$x_3$};
    \node[elliptic state, stuff_fill] (x12) [below left =1cm and -.2cm of x1] {$x_1x_2:=x_1 \vee x_2$};
    \node[elliptic state, stuff_fill] (x13) [below =.85cm of x2] {$x_1x_3:=x_1 \vee x_3$};
    \node[elliptic state] (x23) [below right=1cm and -.2cm of x3] {$x_2x_3:=x_2 \vee x_3$};
    \node[elliptic state, stuff_fill] (x123) [below =1cm of x13] {$x_1x_2x_3 := x_1 \vee x_2 \vee x_3$};
    \node (xml) [left =.1cm of x12]{};
    \node (xmr) [right =.1cm of x23]{};
    \node (t) [above =1.4cm of xml]{d=1};
	\path[->](x1) edge node [align=center]{$$} (x12) 
             (x1) edge node [align=center]{$$} (x13)
             (x2) edge node [align=center]{$$} (x12)
             (x2) edge node [align=center]{$$} (x23)
             (x3) edge node [align=center]{$$} (x13)
             (x3) edge node [align=center]{$$} (x23)
             (x12) edge node [align=center]{$$} (x123)
             (x13) edge node [align=center]{$$} (x123)
             (x23) edge node [align=center]{$$} (x123);
    \path (x123) -- coordinate (l32) (xml) -- coordinate (l21) (x1);
   \draw[dashed](xml |- l32) node[below] {d=3} -- (l32 -| xmr);
   \draw[dashed] (xml |- l21) node[below] {d=2}-- (l21 -| xmr);
\end{tikzpicture}
\caption{Dependencies between different arities of the D-Kernel. The nodes are features and the colored ones are active features. \label{fig:dgraph}}
\end{figure}

So, we expect that the higher the order of the D-Kernel the higher the degree of density of the corresponding kernel matrix. The following theorem and proof, demonstrate that this is true.

\begin{theorem}
For any choice of $d \in \mathbb{N}$, the $d$-degree normalized D-Kernel $\tilde{\kappa}_\vee^d$ is more specific than $\tilde{\kappa}_\vee^{d+1}$, that is $\tilde{\kappa}_\vee^d \leq_G \tilde{\kappa}_\vee^{d+1}$.
	
\end{theorem}

\begin{proof}
	By Lemma \ref{lemma} it is sufficient to prove that
	$$
	\tilde{\kappa}_\vee^d(\mathbf{x}_i,\mathbf{x}_j)^2 \leq \tilde{\kappa}_\vee^{d+1}(\mathbf{x}_i,\mathbf{x}_j)^2
	$$
	is true.
	Let us consider the left hand side term of the above mentioned inequality:
	$$
		\tilde{\kappa}_\vee^d(\mathbf{x}_i,\mathbf{x}_j)^2 = \frac{\kappa_\vee^d(\mathbf{x}_i,\mathbf{x}_j)^2}{\kappa_\vee^d(\mathbf{x}_i,\mathbf{x}_i) \kappa_\vee^d(\mathbf{x}_j,\mathbf{x}_j)} = \frac{\kappa_\vee^d(\mathbf{x}_i,\mathbf{x}_j)}{\kappa_\vee^d(\mathbf{x}_i,\mathbf{x}_i)} \frac{\kappa_\vee^d(\mathbf{x}_i,\mathbf{x}_j)}{\kappa_\vee^d(\mathbf{x}_j,\mathbf{x}_j)}.
	$$
	We can write the term $\frac{\kappa_\vee^d(\mathbf{x}_i,\mathbf{x}_j)}{\kappa_\vee^d(\mathbf{x}_i,\mathbf{x}_i)}$ as:
	\begin{align*}
		\frac{\kappa_\vee^d(\mathbf{x}_i,\mathbf{x}_j)}{\kappa_\vee^d(\mathbf{x}_i,\mathbf{x}_i)} &= \frac{N_d(\mathbf{f}) - N_d(\overline{\mathbf{x}_i}) - N_d(\overline{\mathbf{x}_j}) + N_d(\overline{\mathbf{x}_i \cup \mathbf{x}_j})}{N_d(\mathbf{f}) - N_d(\overline{\mathbf{x}_i})}\\ 
		&= 1 - \frac{N_d(\overline{\mathbf{x}_j}) - N_d(\overline{\mathbf{x}_i \cup \mathbf{x}_j})}{N_d(\mathbf{f}) - N_d(\overline{\mathbf{x}_i})}.
	\end{align*}
	
	Let us reflect on the meaning of the last fraction. At the numerator we have $N_d(\overline{\mathbf{x}_j}) - N_d(\overline{\mathbf{x}_i \cup \mathbf{x}_j})$ which counts the number of $d$-combinations without repetition which have at least one active variable in $\mathbf{x}_i$ but that is not in $\mathbf{x}_j$. Similarly, the denominator, that is $N_d(\mathbf{f}) - N_d(\overline{\mathbf{x}_i})$, counts the $d$-combinations without repetition which have at least one active variable in $\mathbf{x}_i$.
	Formally, we have:
	\begin{align*}
		N_d(\overline{\mathbf{x}_j}) - N_d(\overline{\mathbf{x}_i \cup \mathbf{x}_j}) &= (n_i - n_{ij})N_{d-1}(\overline{\mathbf{x}_j}\setminus \{a_j\})\\
		N_d(\mathbf{f}) - N_d(\overline{\mathbf{x}_i}) &= n_i N_{d-1}(\mathbf{f}\setminus \{a_i\}),
	\end{align*}
	where $a_i \in \mathbf{x}_i$ and $a_j \in \mathbf{x}_j$ are generic active variables in $\mathbf{x}_i$ and $\mathbf{x}_j$, respectively.
	
	So, we have that
	\begin{align*}
	\frac{\kappa_\vee^d(\mathbf{x}_i,\mathbf{x}_j)}{\kappa_\vee^d(\mathbf{x}_i,\mathbf{x}_i)} &= 1-\frac{(n_i - n_{ij})}{n_i} \frac{N_{d-1}(\overline{\mathbf{x}_j}\setminus \{a_j\})}{N_{d-1}(\mathbf{f}\setminus \{a_i\})} = 1-\frac{(n_i - n_{ij})}{n_i} \frac{\binom{n-n_j-1}{d-1}}{\binom{n-1}{d-1}}\\
	&= 1-\frac{(n_i - n_{ij})}{n_i} \frac{\frac{(n-n_j-1)!}{(d-1)!(n-n_j-d)!}}{\frac{(n-1)!}{(d-1)!(n-d)!}} \\
	&= 1-\frac{(n_i - n_{ij})}{n_i} \frac{(n-n_j-1)!}{(n-1)!} \frac{(n-d)!}{(n-n_j-d)!} \\
	&= 1-\beta \frac{(n-d)(n-d-1)\dots(n-n_j-d+1)(n-n_j-d)!}{(n-n_j-d)!}\\
	&= 1-\beta [(n-d)(n-d-1)\dots(n-n_j-d+1)],
	\end{align*}
	where $\beta = \frac{(n_i - n_{ij})}{n_i} \frac{(n-n_j-1)!}{(n-1)!}$. 
	
	Thus, it is easy to see that
	\begin{align*}
	& \frac{\kappa_\vee^d(\mathbf{x}_i,\mathbf{x}_j)}{\kappa_\vee^d(\mathbf{x}_i,\mathbf{x}_i)} = 1-\beta [(n-d)(n-d-1)\dots(n-n_j-d+1)] \leq\\
	& 1-\beta [(n-d-1)(n-d-2)\dots(n-n_j-d)] = \frac{\kappa_\vee^{d+1}(\mathbf{x}_i,\mathbf{x}_j)}{\kappa_\vee^{d+1}(\mathbf{x}_i,\mathbf{x}_i)}	.
	\end{align*}
	Since the same consideration can be made for $\frac{\kappa_\vee^d(\mathbf{x}_i,\mathbf{x}_j)}{\kappa_\vee^d(\mathbf{x}_j,\mathbf{x}_j)}$, we can conclude that
	$$
	\tilde{\kappa}_\vee^d(\mathbf{x}_i,\mathbf{x}_j)^2 = \frac{\kappa_\vee^d(\mathbf{x}_i,\mathbf{x}_j)}{\kappa_\vee^d(\mathbf{x}_i,\mathbf{x}_i)} \frac{\kappa_\vee^d(\mathbf{x}_i,\mathbf{x}_j)}{\kappa_\vee^d(\mathbf{x}_j,\mathbf{x}_j)} \leq \frac{\kappa_\vee^{d+1}(\mathbf{x}_i,\mathbf{x}_j)}{\kappa_\vee^{d+1}(\mathbf{x}_i,\mathbf{x}_i)} \frac{\kappa_\vee^{d+1}(\mathbf{x}_i,\mathbf{x}_j)}{\kappa_\vee^{d+1}(\mathbf{x}_j,\mathbf{x}_j)} = \tilde{\kappa}_\vee^{d+1}(\mathbf{x}_i,\mathbf{x}_j)^2.
	$$
\end{proof}

In Section \ref{sec:exp} it is empirically shown, for each dataset, how the normalized spectral ratio decreases with the increasing of the degree.

\section{Experiments and Results}\label{sec:exp}
In this section we present the extensive experimental work we have done. All the experiments are based on the \textit{CF-KOMD} framework. The implementation of the framework and all the used datasets are available at \url{https://github.com/makgyver/pyros}.

\subsection{Datasets}\label{sec:datasets}
In this section we introduce the datasets used in the experiments. Table \ref{tab:datasets} shows 
a brief description of the datasets.

\begin{savenotes}
\begin{table}[h]
	\centering
	\begin{tabular}{|c|c|c|c|c|c|}
		\hline
		\textbf{Dataset} & \textbf{Item type} & $|\mathcal{U}|$ & $|\mathcal{I}|$ & $|\mathcal{R}|$ & \textbf{Density}\\ \hline
		MovieLens\footnote{\url{http://grouplens.org/datasets/movielens}} & Movies & 6040 & 3706 & 1M & 1.34\% \\ \hline
		BookCrossing \footnote{Reduced version of the dataset \url{http://grouplens.org/datasets/book-crossing}}& Books & 2802 & 5892  & 70593 & 0.004\%\\ \hline
		Ciao \cite{Guo:2014}\footnote{\url{http://www.librec.net/datasets.html#ciaodvd}}& Movies & 17615 & 16121 & 72664 & 0.025\%\\ \hline
		Netflix \footnote{Reduced version of the Netflix Prize Dataset \url{http://www.netflixprize.com}} & Movies & 93705 & 3561 & 3.3M & 0.99\%\\ \hline
		FilmTrust \cite{Guo:2013}\footnote{\url{http://www.librec.net/datasets.html#filmtrust}} & Movies & 1508 & 2071 & 35496 & 1.13\%\\ \hline
		Jester \cite{Goldberg:2001}\footnote{\url{http://goldberg.berkeley.edu/jester-data/}} & Jokes & 24430 & 100 & 1M & 42.24\%\\ \hline
	\end{tabular}
	\caption{Brief description of the used datasets.\label{tab:datasets}}
\end{table}
\end{savenotes}

One of the main characteristic of these datasets is the high sparsity of the ratings, with the exception of \verb#Jester#, and how these ratings are distributed. In particular, as extensively analyzed in \cite{Polato:2017}, the distribution of the ratings on the items is long tailed. While, from the users point of view, only the \verb#Ciao# and \verb#Book Crossing# datasets have a well defined long tail distribution.

For experimental purposes, we removed from the \verb#Jester# dataset all the users with more than 90 ratings (i.e., $>90\%$ of the items).

\subsection{Experimental setting}
Experiments have been replicated 5 times for each dataset. Datasets have been pre-processed as described in the following:
\begin{enumerate}
	\item users are randomly split in 5 sets of the same size;
	\item for each user its ratings are further split in two halves;
	\item at each round test, we use all the ratings in 4 sets of users plus one half of ratings of the remaining set as training, and the remaining ratings as test set;
	\item users with less than 5 ratings are forced to be in the training set.
\end{enumerate}
This setting avoids situations of users cold start, since in the training phase we have at least five ratings for every user in the test set. 
Results reported below are the averages (with its standard deviations) over the 5 folds.\\

The rankings' evaluation metric used to compare the performances of the methods is the AUC (Area Under the receiver operating characteristic Curve) defined as in the following:
$$
	\textit{AUC} = \frac{1}{|\mathcal{U}|} \sum\limits_{u\in\mathcal{U}} \frac{1}{|\mathcal{I}_u| \cdot |\mathcal{I}\setminus \mathcal{I}_u|} \sum\limits_{i \in \mathcal{I}_u} \sum\limits_{j \notin \mathcal{I}_u} \mathbb{I}[\hat{r}_{ui} > \hat{r}_{uj}]
$$
where $\hat{r}_{ui}, \hat{r}_{uj}$ are the predicted scores for the user-item pairs $(u,i), (u,j)$ and $\mathbb{I} : \textit{Bool} \rightarrow \mathbb{B}$ is the indicator function which returns 1 if the predicate is \textit{true} and 0 otherwise.

\subsection{Top-N recommendation}
In this section, we show the results obtained on the datasets described in Section \ref{sec:datasets} using the CF-KOMD framework (see Section \ref{sec:cfkomd}). We compare the D-Kernel and the C-Kernel against the best perfoming kernels on these datasets as reported in \cite{Polato:2017}. 
In all the experiments the parameter $\lambda_p$ of the optimization problem (\ref{opt3}) is fixed to 0.1 (different values did not make much difference as shown in \cite{Polato:2016}).

Table \ref{tab:results} summarizes the results. All the reported results are in terms of AUC. We also measure the performances with the evaluation metrics mAP@10 (Mean Average Precision), and nDCG@10 (Normalized Discounted Cumulative Gain), and they are consistent with the AUC.

\begin{table}[h!]
	\centering
	{\renewcommand{\arraystretch}{1.1}
	\begin{tabular}{r|c|c|c|c|c|}
		&\textbf{D-Kernel} & \textbf{C-Kernel} & \textbf{Linear} & \textbf{Tanimoto} & \textbf{mDNF}\\ \hline
		
		\verb#BookCrossing# & $\mathbf{\underset{\pm 0.0027}{0.7651}} (8)$ & $\underset{\pm 0.0023}{0.7040} (2)$ & $\underset{\pm 0.0023}{0.7439}$ & $\underset{\pm 0.0021}{0.7592}$ & $\underset{\pm 0.0017}{0.7371} $ \\
		
		\verb#Ciao# & $\mathbf{\underset{\pm 0.0065}{0.8381}} (116)$ & $\underset{\pm 0.0094}{0.5401} (2)$ & $\underset{\pm 0.0057}{0.7179}$ & $\underset{\pm 0.0041}{0.7538} $ & $\underset{\pm 0.004}{0.7859} $ \\
		
		\verb#FilmTrust# & $\mathbf{\underset{\pm 0.0036}{0.9705}} (38)$ & $\underset{\pm 0.0075}{0.9511} (2)$ & $\underset{\pm 0.0049}{0.9611}$ & $\underset{\pm 0.0054}{0.964} $ & $\underset{\pm 0.0056}{0.9608} $ \\		
		
		\verb#Jester# & $\mathbf{\underset{\pm 0.0024}{0.7271}} (34)$ & $\underset{\pm 0.0028}{0.6012} (2)$ & $\underset{\pm 0.0026}{0.6019}$ & $\underset{\pm 0.0025}{0.6263} $ & $\underset{\pm 0.0981}{0.5384}$ \\	

		\verb#MovieLens# &  $\underset{\pm 0.0004}{0.8946} (2)$ & $\underset{\pm 0.0004}{0.8923} (2)$ & $\mathbf{\underset{\pm 0.0003}{0.8961}}$ & $\underset{\pm 0.0004}{0.8922} $ & $\underset{\pm 0.0027}{0.8065}$ \\	

		\verb#Netflix# & $\mathbf{\underset{\pm 0.001}{0.9436}} (2)$ & $\underset{\pm 0.0009}{0.9362} (2)$ & $\underset{\pm 0.0011}{0.9427}$ & $\underset{\pm 0.0012}{0.9381} $ & $\underset{\pm 0.0013}{0.9}$ \\	 \hline
		
	\end{tabular}}
	\caption{AUC results on five CF datasets using different kernels. Results are reported with their standard deviation and the best performing degrees (when applicable) are indicated inside the parenthesis. For each dataset the highest AUC is highlighted in \textbf{bold}. \label{tab:results}}
\end{table}

In almost all the datasets (5 out of 6) the D-Kernel achieves the best AUC score with a significant improvement on \verb#Ciao# and \verb#Jester#. In fact, this result is higher than the state-of-the-art result achieved by MSDW (see \cite{Polato:2017}). It is also worth to notice that the best score on \verb#Ciao# is reached with $d=116$ which is surprisingly high. However, this is due to the combination of some peculiar features of this dataset, such as, its high degree of sparseness, its high number of features (i.e., users) and its double long tailed distribution. These characteristics entail a very slow decrease of the complexity with the increasing of $d$, as underlined by the Figure \ref{fig:comp-ciao}.
Conversely, for both \verb#Netflix# and \verb#MovieLens# the best performing arity for the D-Kernel is 2, even though, in the first case it is also the overall highest score, while, in the second, the best score is achieved by the linear kernel. 

We also evaluated the performance of the C-Kernel and it achieves poor performance in all datasets. This is due to its high expressiveness: in all the cases its best results have been reached with $d=2$. It is important to underline that the null vector problem described in Section \ref{sec:ckernel} has artificially been solved by forcing, in the kernel matrix, 1 in correspondence of the diagonal of these degenerate examples. Even though this is a reasonable fix, it is an arbitrary choice.

Regarding the mDNF kernel, results confirm that too high expressiveness do not lead to good performances. Specifically, in half of the datasets the mDNF kernel performs better than the C-Kernel and it is comparable to the linear one. On the other half its performances are very poor due to its high complexity (see Figure \ref{fig:compc}).

In Figure \ref{fig:dres} and \ref{fig:cres}  the AUC score achieved with different arities of the D-Kernel and the C-Kernel, respectively, are depicted. In both cases the trend of the plots are quite consistent. The performances of the D-Kernel slowly increase until they reach a peak and then they start to constantly decrease. In some cases the peak is reached with the smallest (non trivial) arity $d=2$ and so the plot is always decreasing. 

In contrast, the C-Kernel's performances rapidly decrease since they become very poor. This is a further confirmation that these kind of datasets need kernel with a  complexity lower than the linear one.

\begin{figure}[htbp]
\centering
\subfigure[BookCrossing]
{\includegraphics[width=3.95cm]{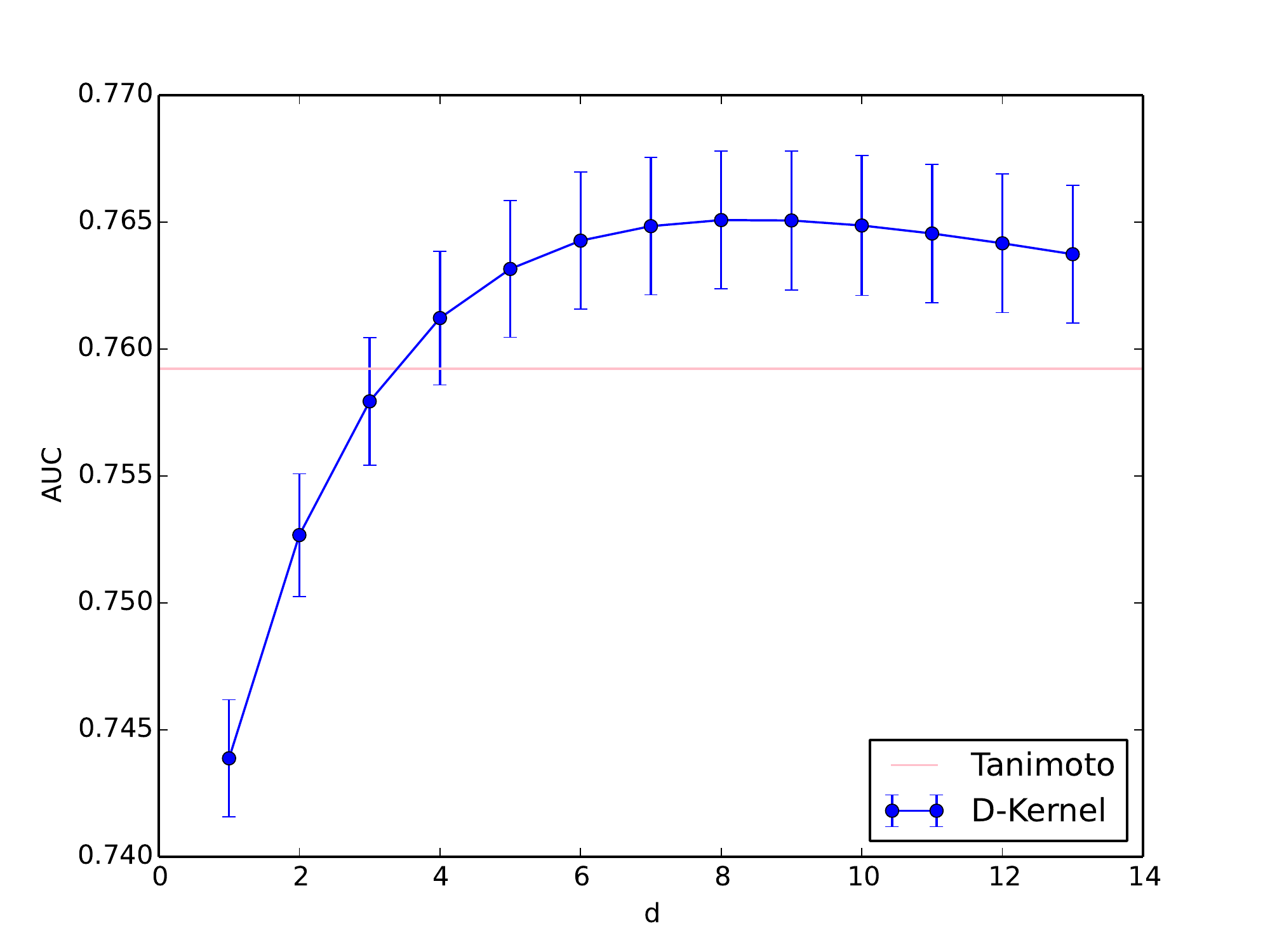}}
\subfigure[Ciao]
{\includegraphics[width=3.95cm]{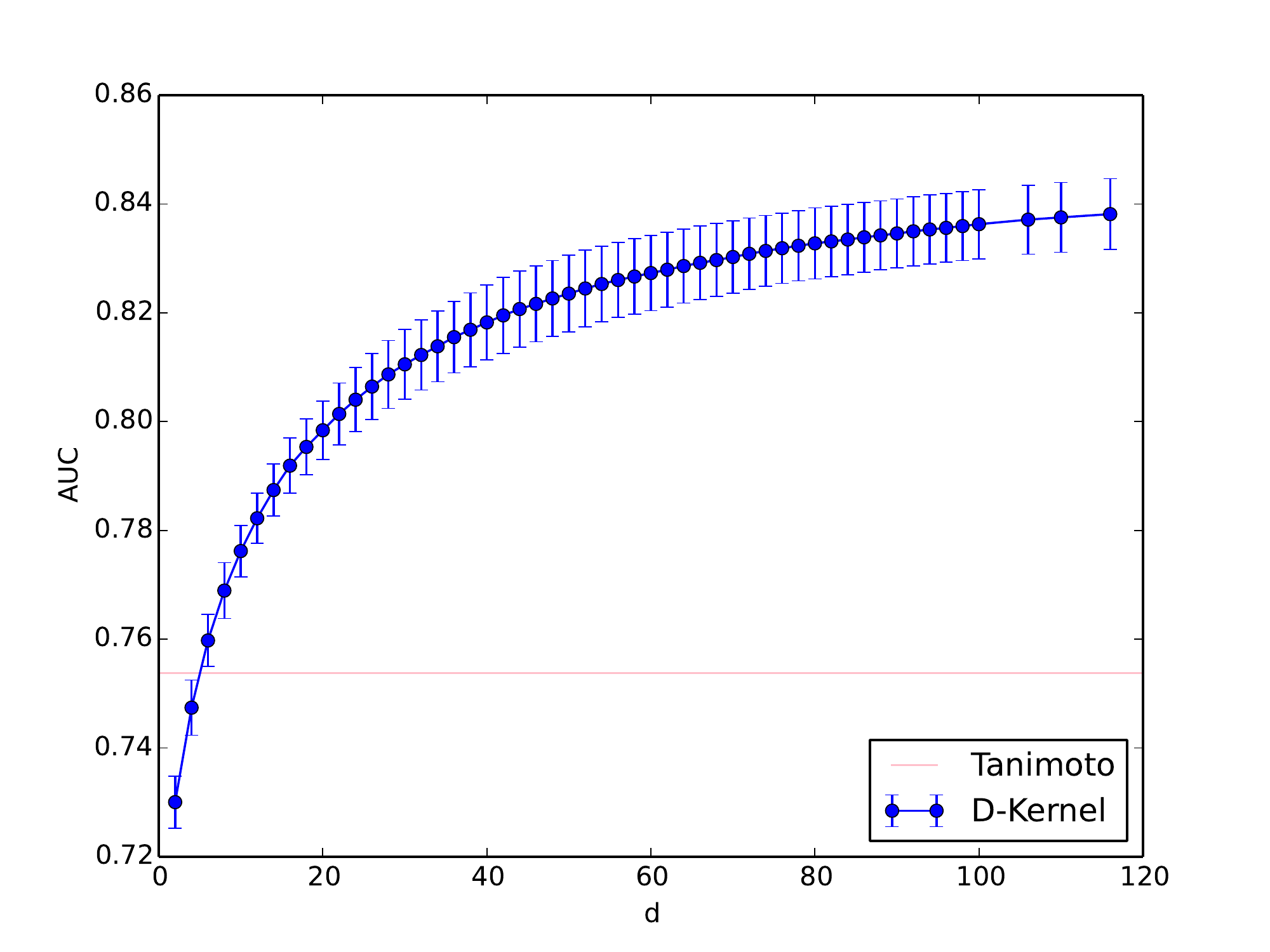}}
\subfigure[Film Trust]
{\includegraphics[width=3.95cm]{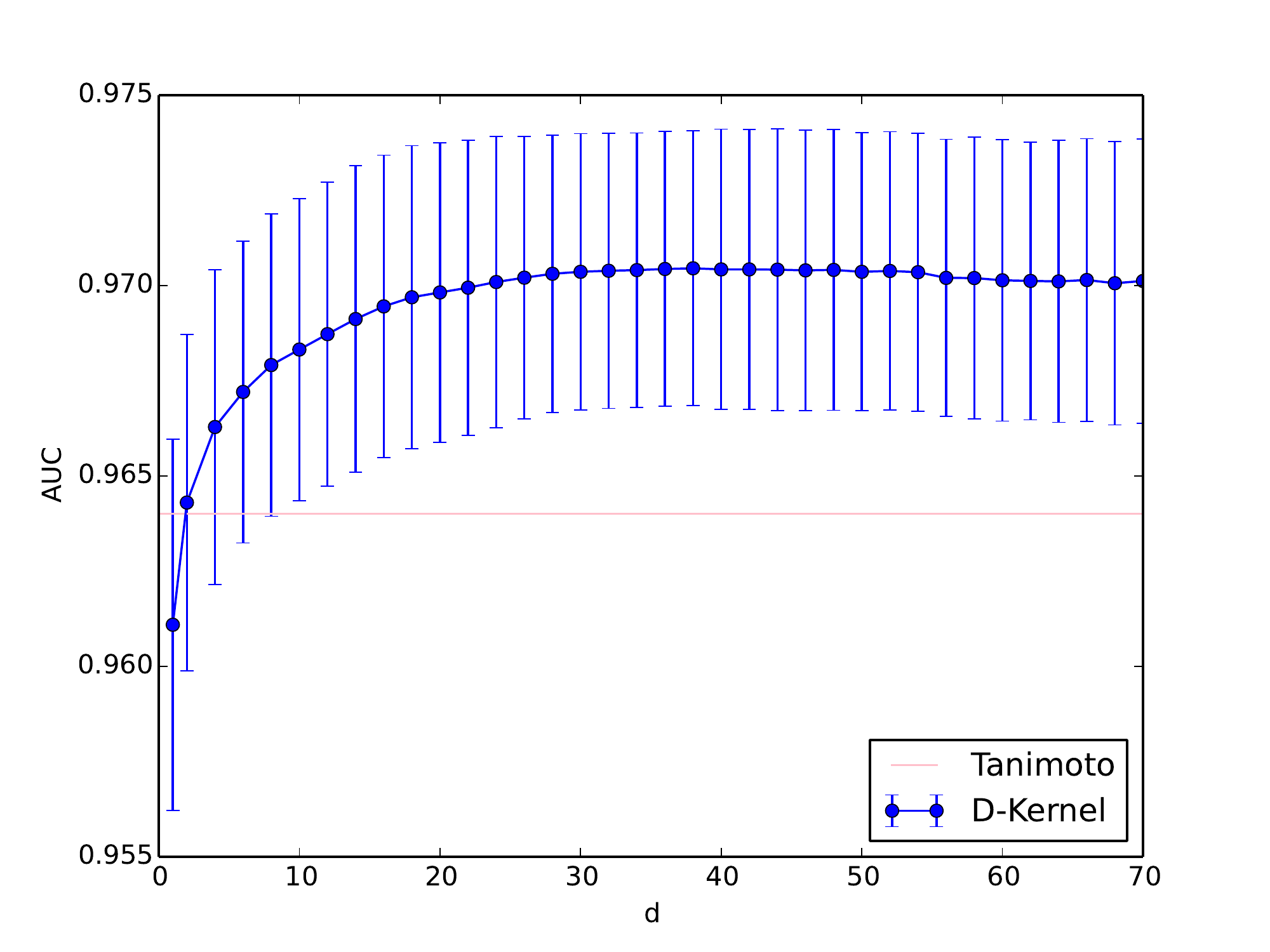}}
\subfigure[Jester]
{\includegraphics[width=3.95cm]{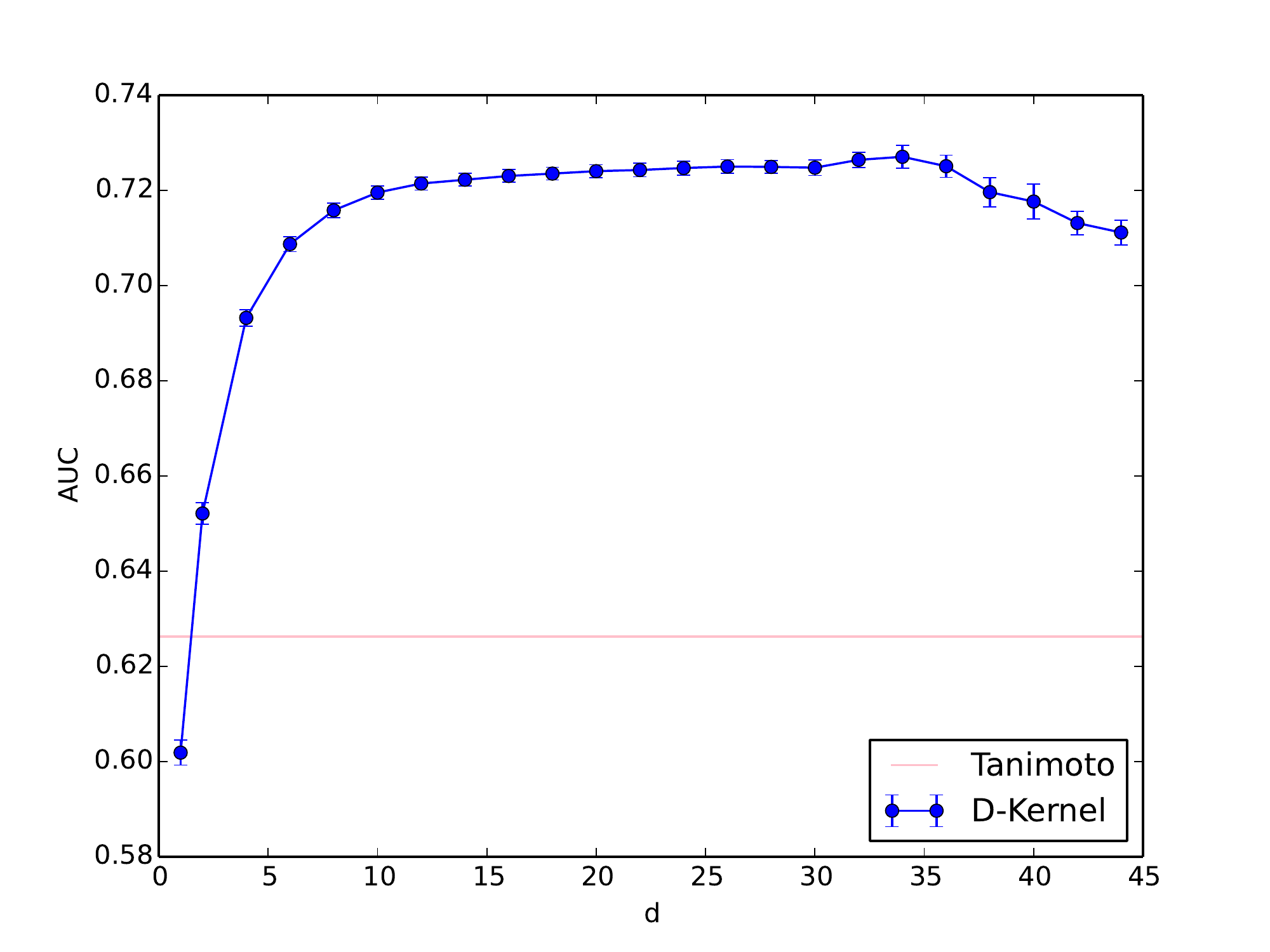}}
\subfigure[MovieLens]
{\includegraphics[width=3.95cm]{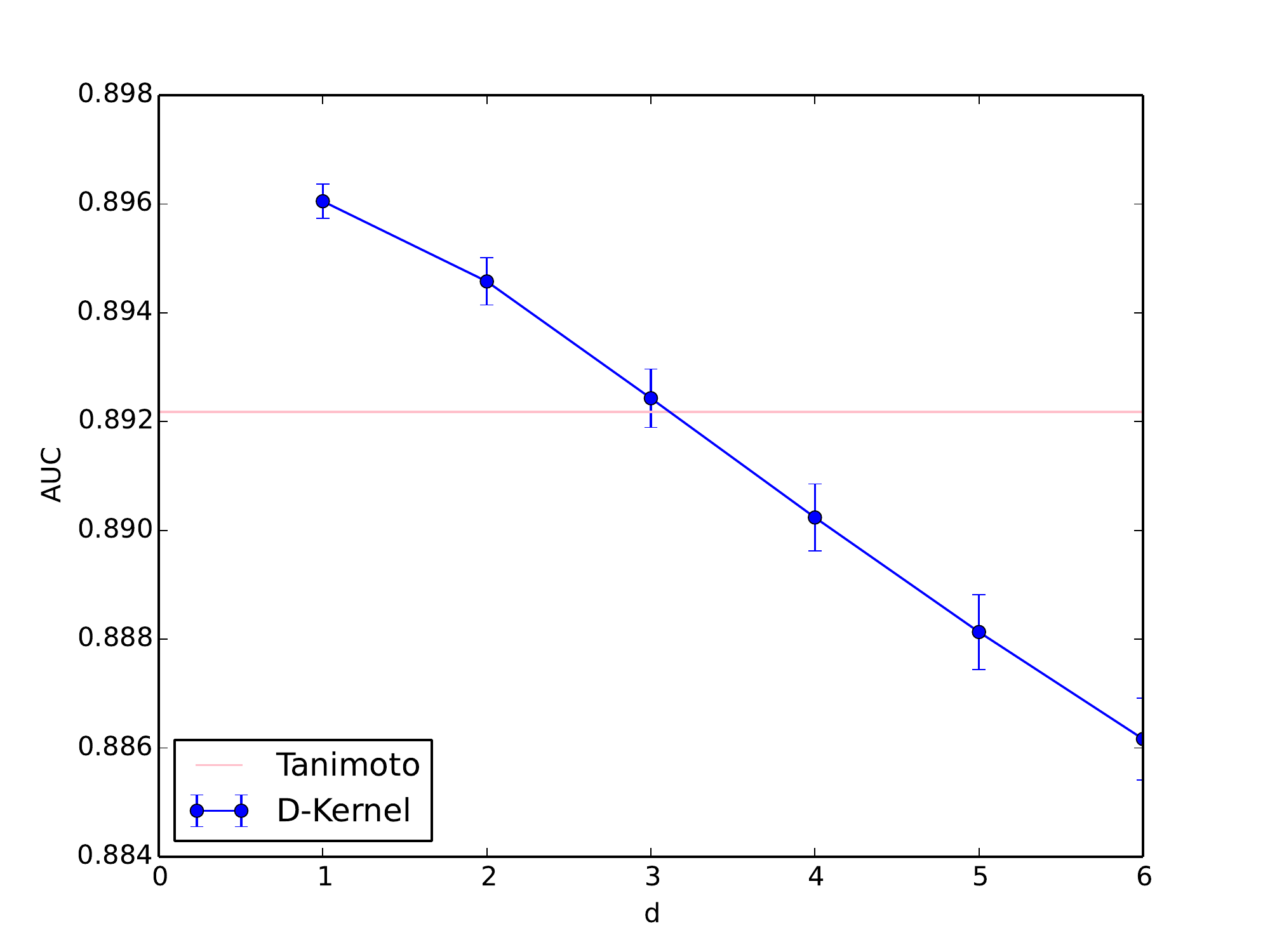}}
\subfigure[Netflix]
{\includegraphics[width=3.95cm]{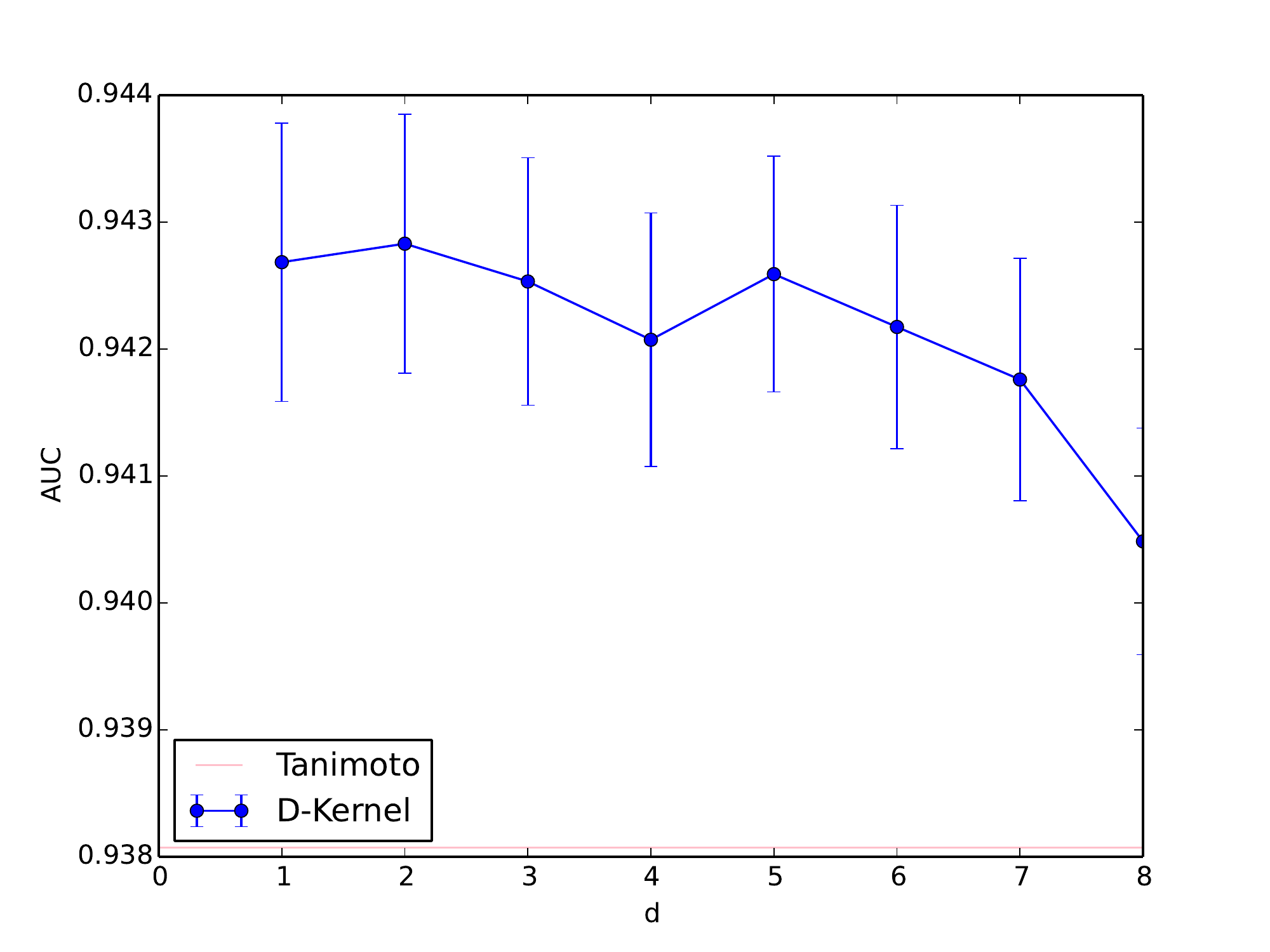}}
\caption{Performance of different D-Kernel degrees.\label{fig:dres}}
\end{figure}

\begin{figure}[h!]
\centering
\subfigure[BookCrossing]
{\includegraphics[width=3.95cm]{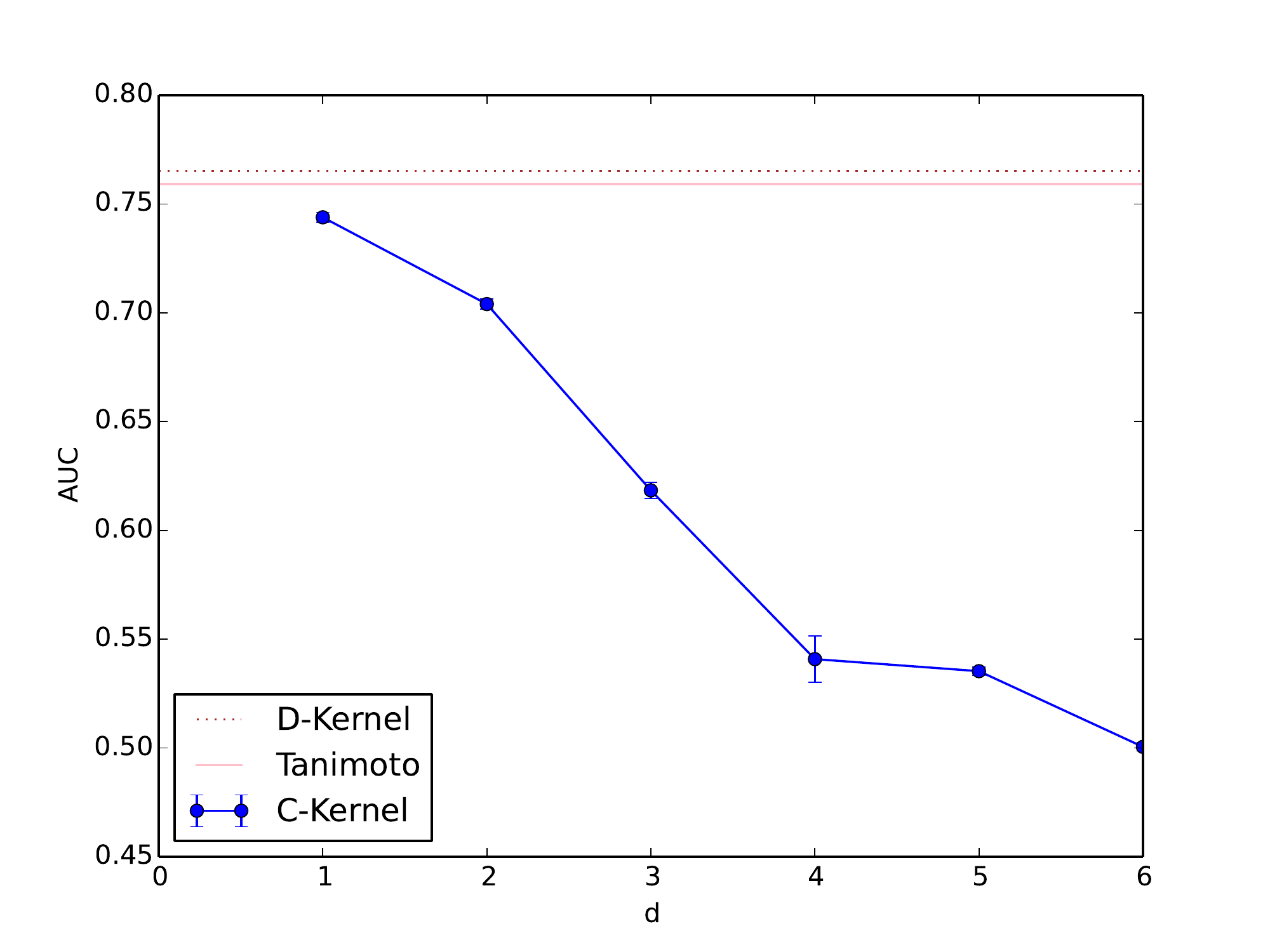}}
\subfigure[Ciao]
{\includegraphics[width=3.95cm]{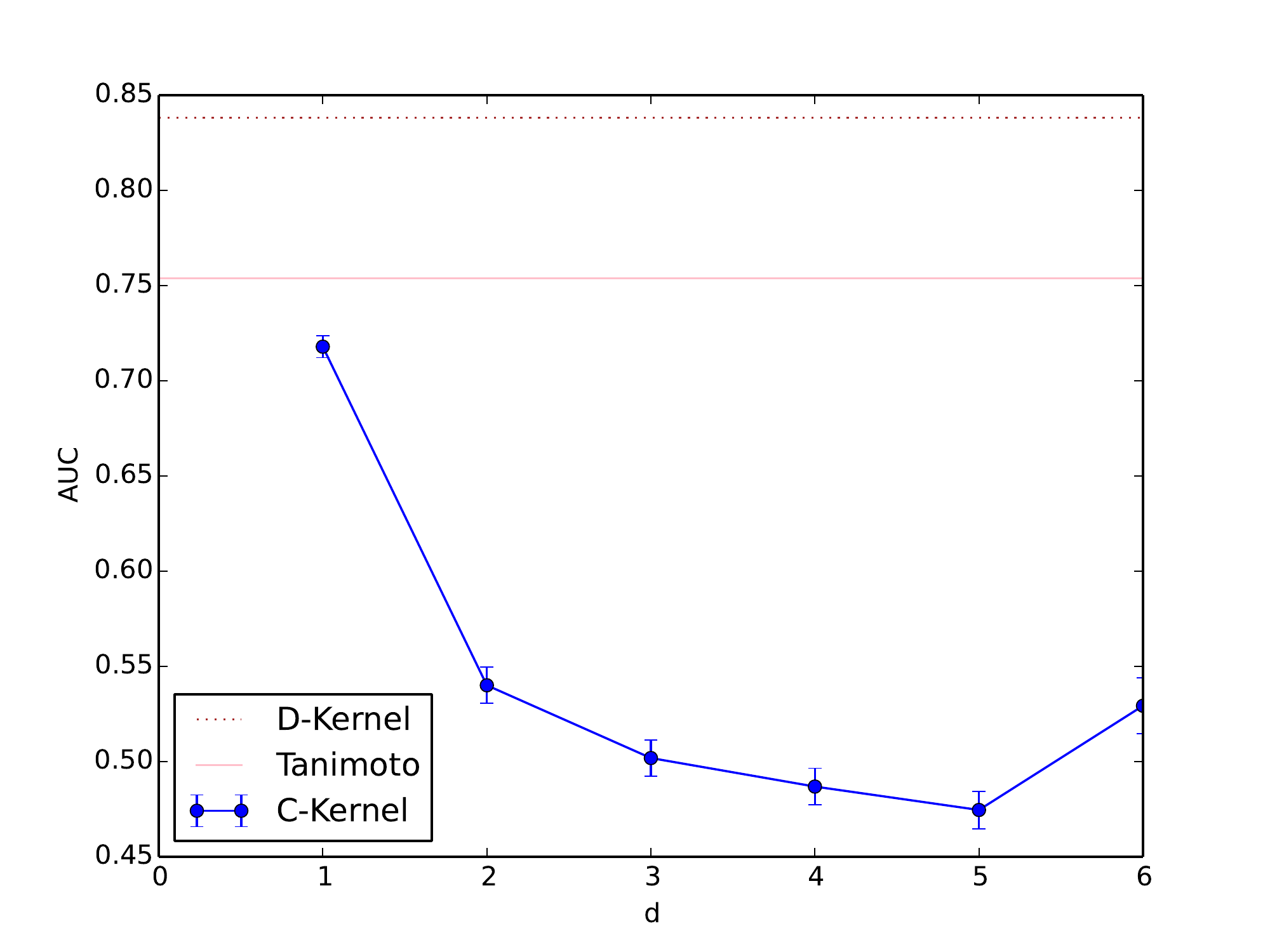}}
\subfigure[Film Trust]
{\includegraphics[width=3.95cm]{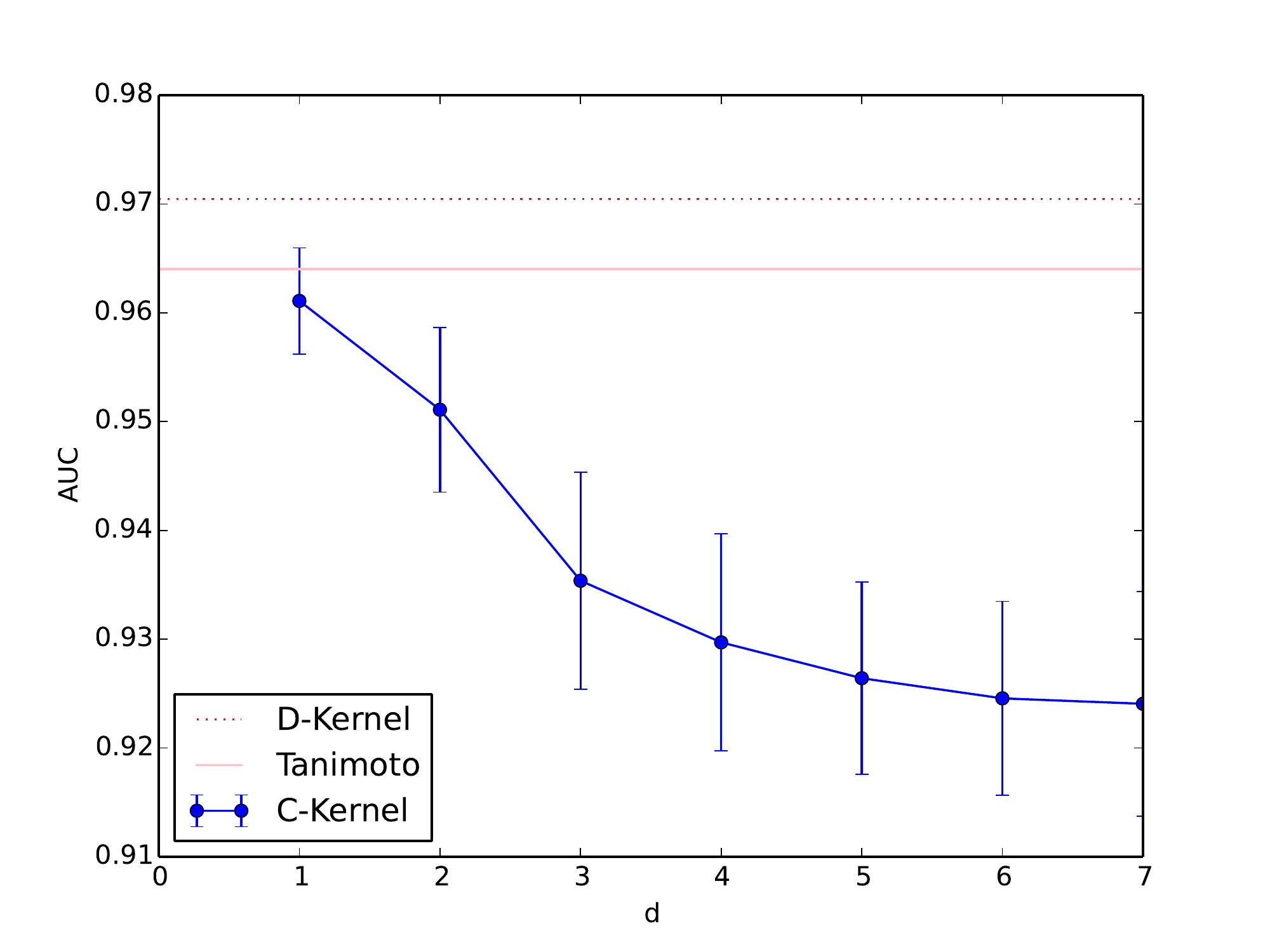}}
\subfigure[Jester]
{\includegraphics[width=3.95cm]{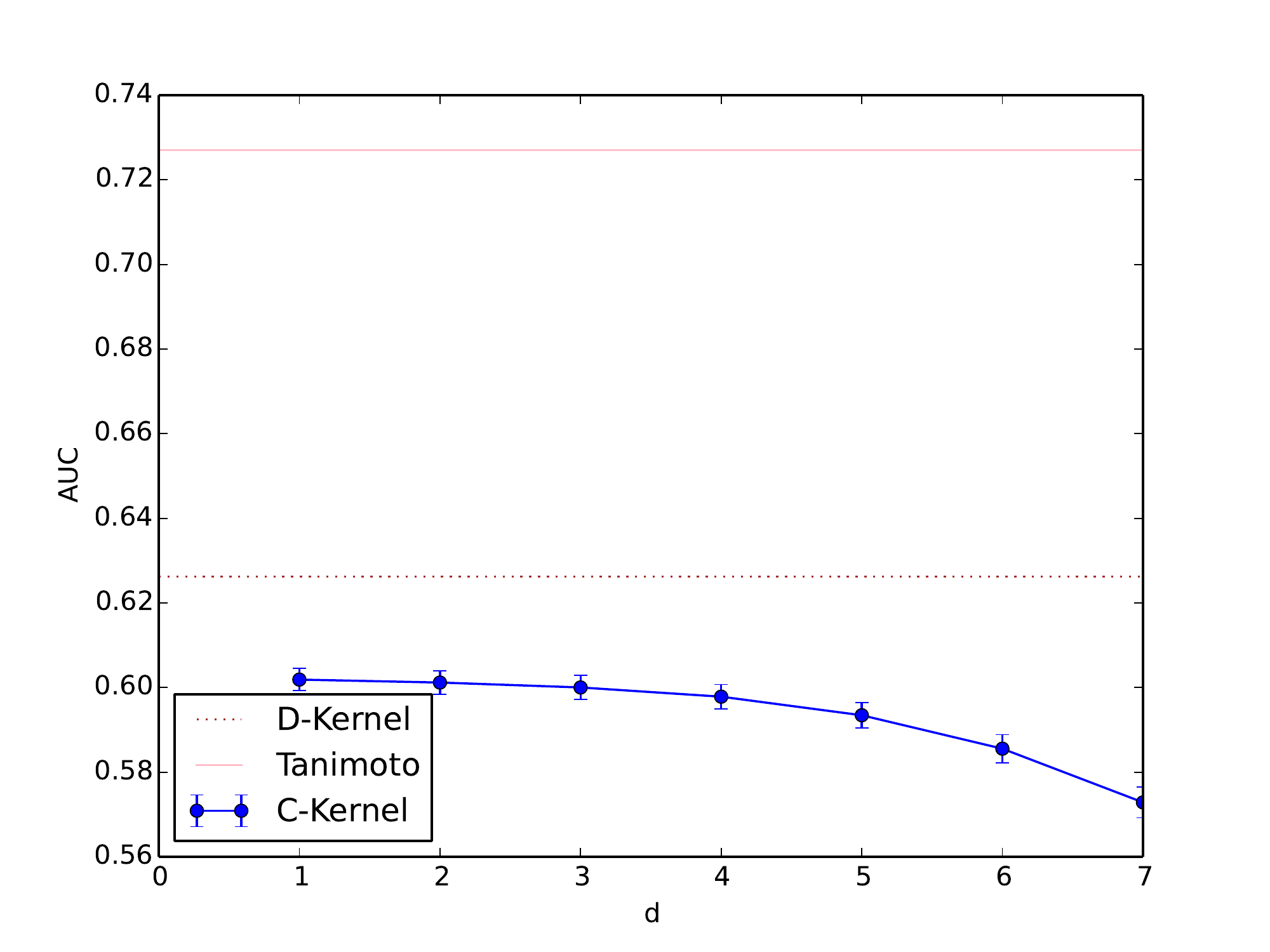}}
\subfigure[MovieLens]
{\includegraphics[width=3.95cm]{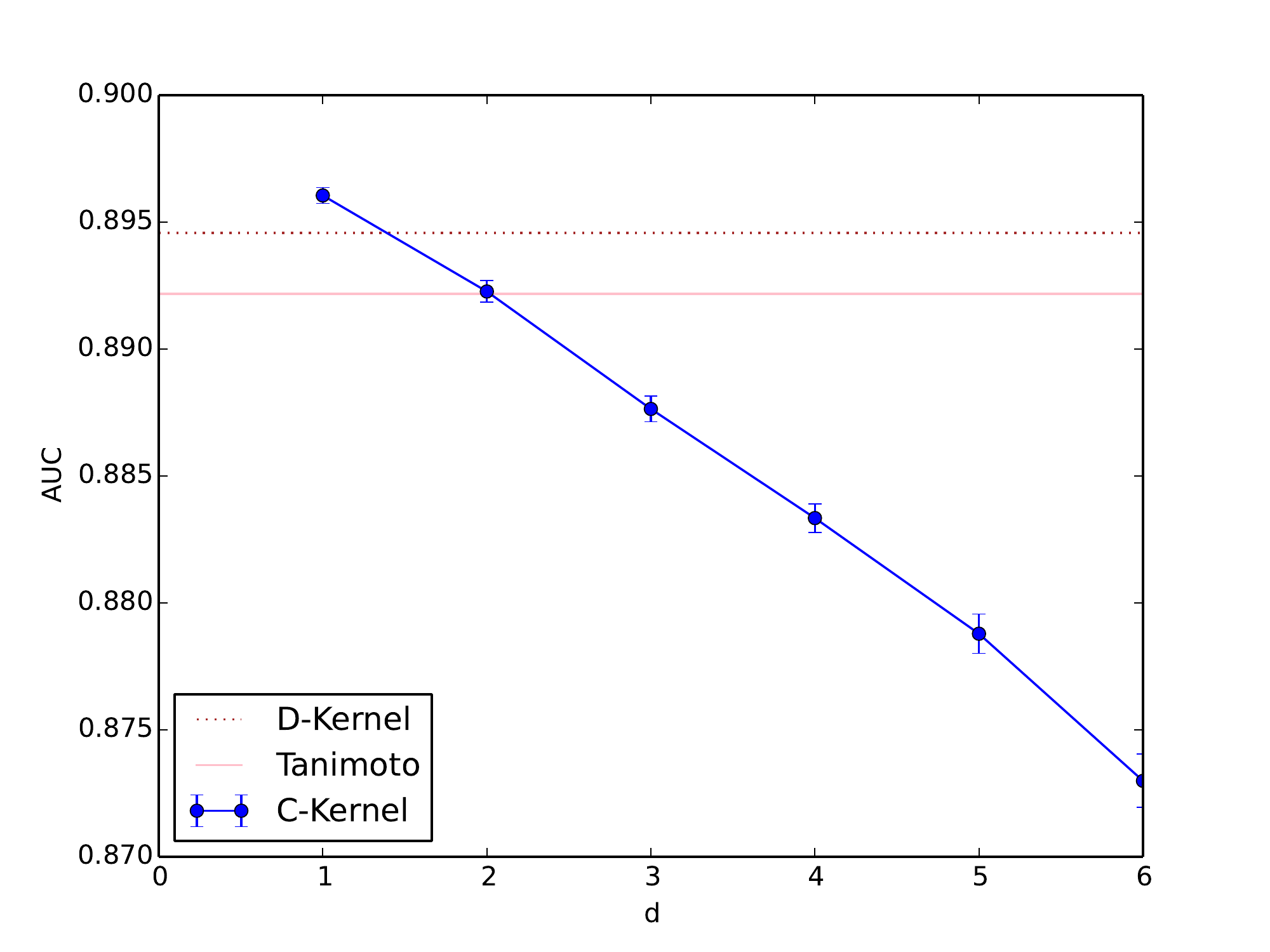}}
\subfigure[Netflix]
{\includegraphics[width=3.95cm]{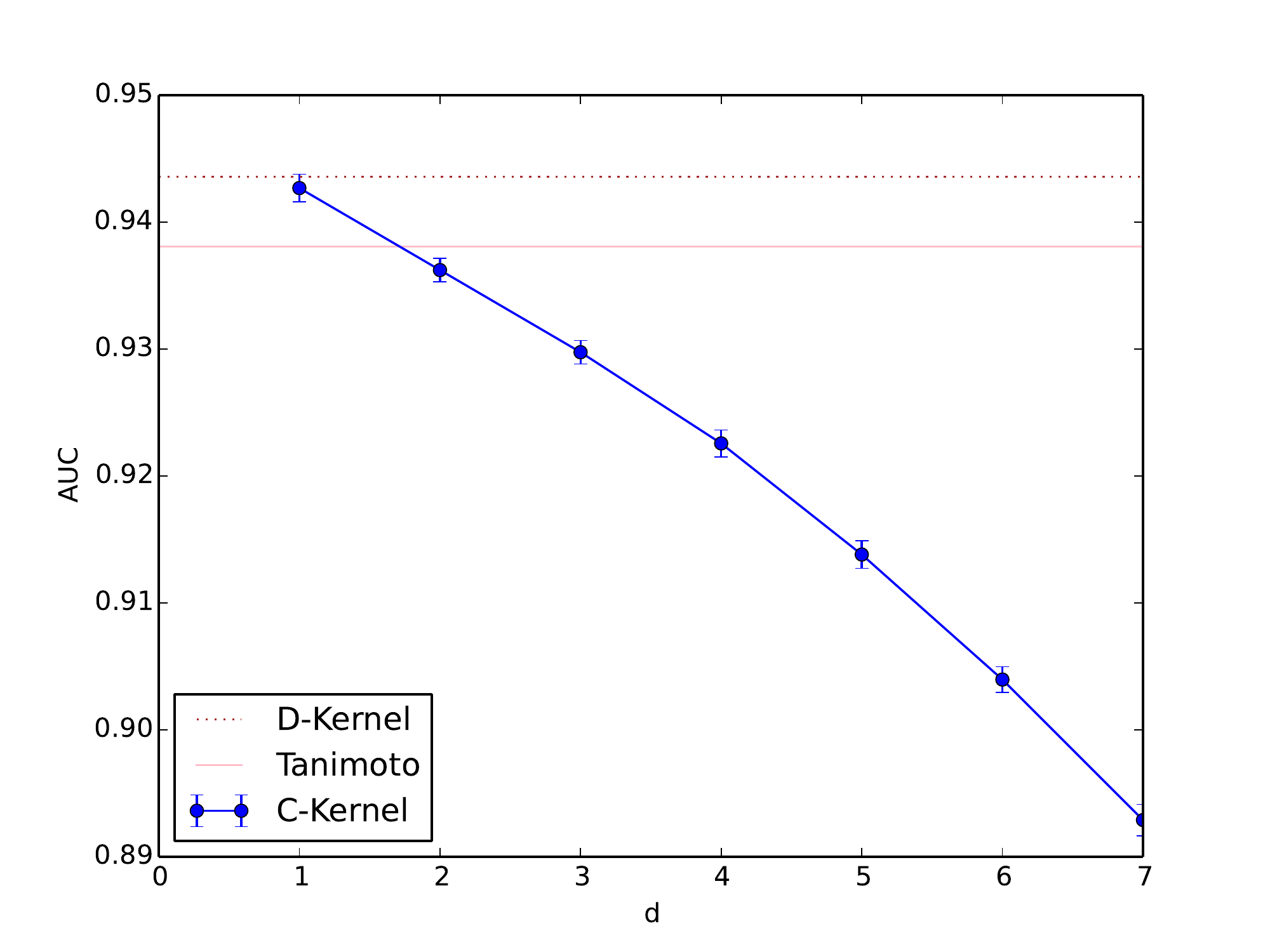}}
\caption{Performance of different C-Kernel arities.\label{fig:cres}}
\end{figure}

\subsection{Spectral ratio}
In Section \ref{sec:exp-ck} and \ref{sec:exp-dk} we have proved how the spectral ratio is linked with the arity  of the C-Kernel and the D-Kernel. Figure \ref{fig:compc} and \ref{fig:compd} confirm our theoretical results.

\begin{figure}[h!]
\centering
\subfigure[BookCrossing]
{\includegraphics[width=6cm]{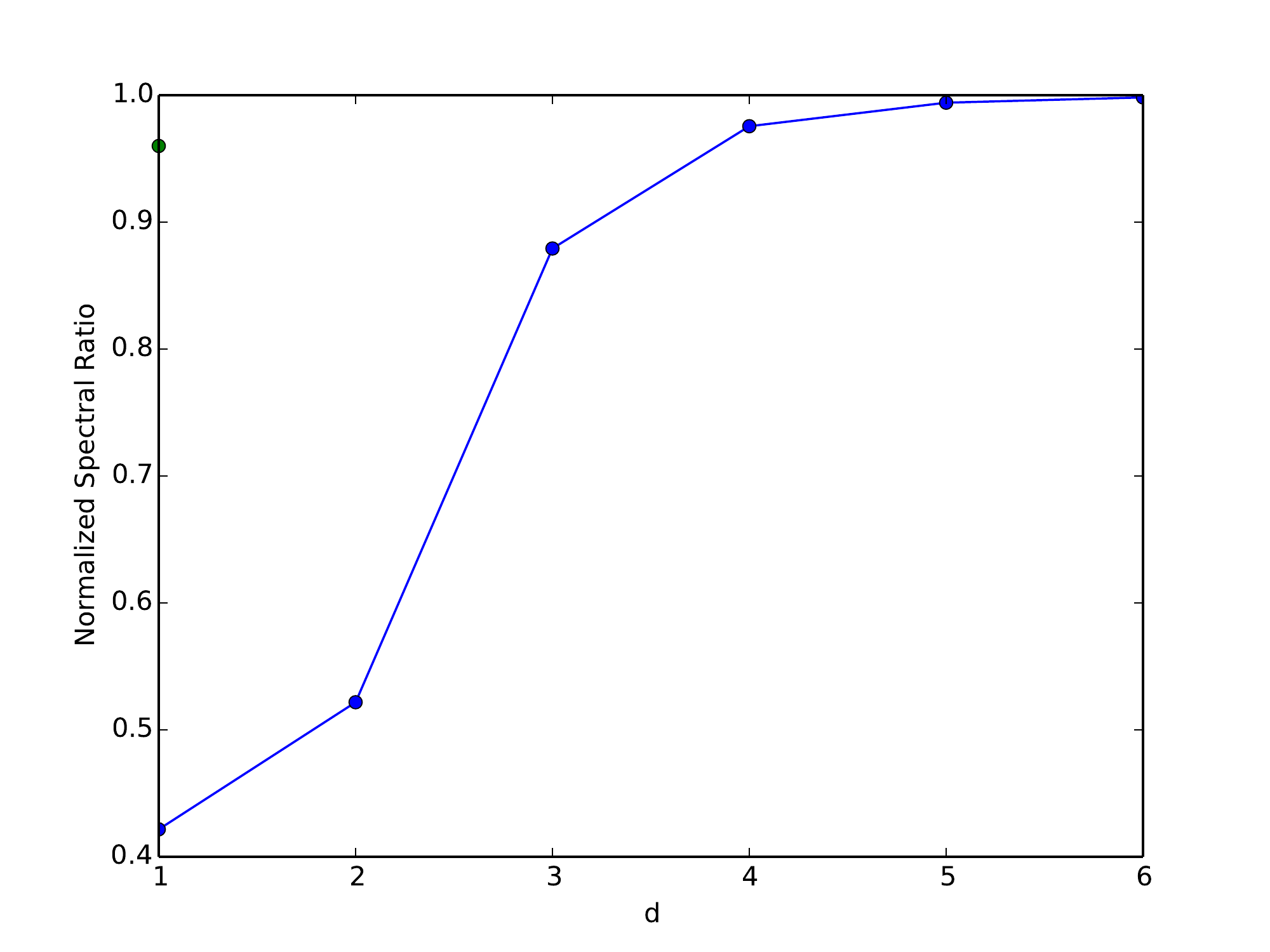}}
\subfigure[Ciao]
{\includegraphics[width=6cm]{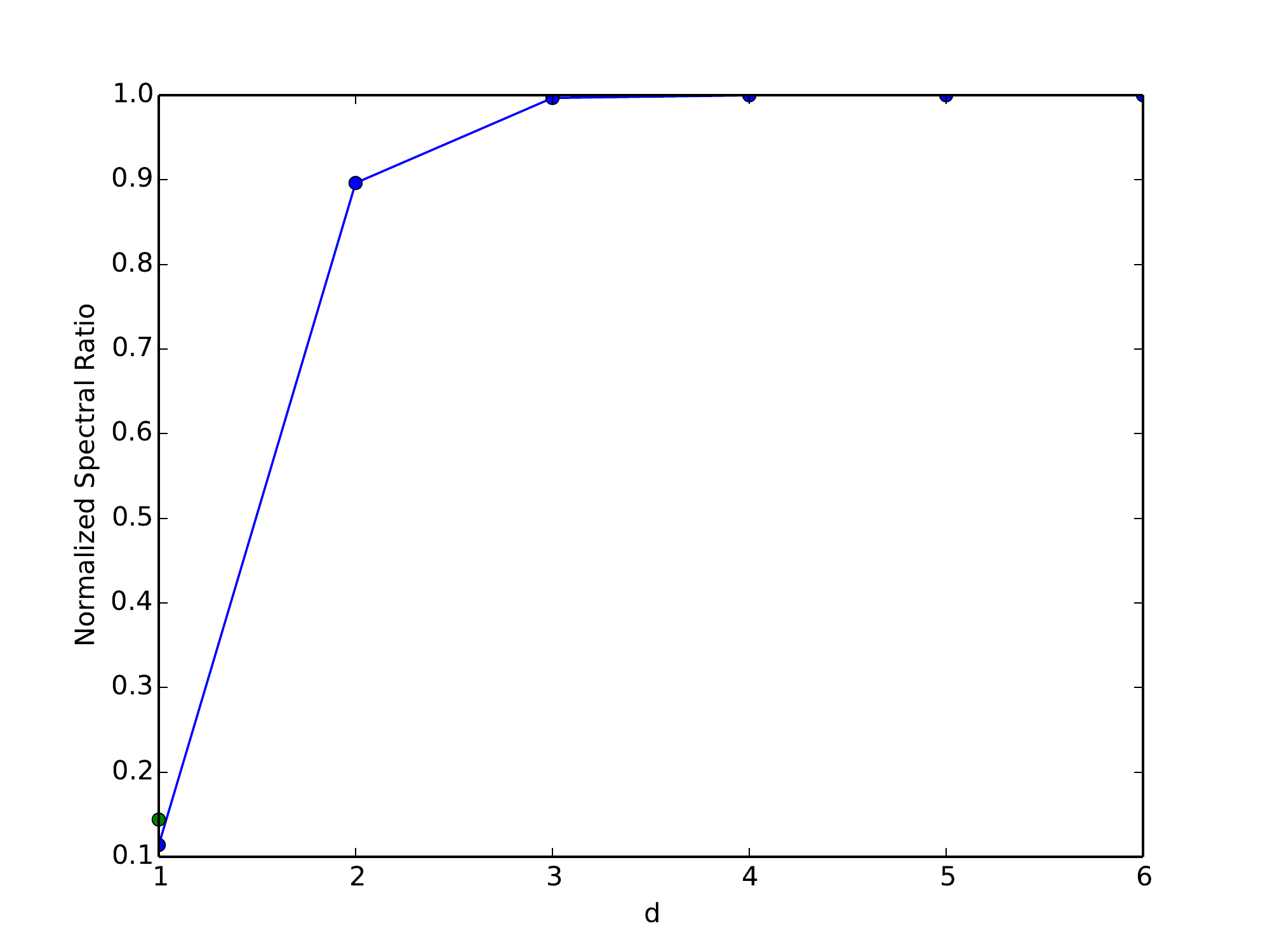}}
\subfigure[Film Trust]
{\includegraphics[width=6cm]{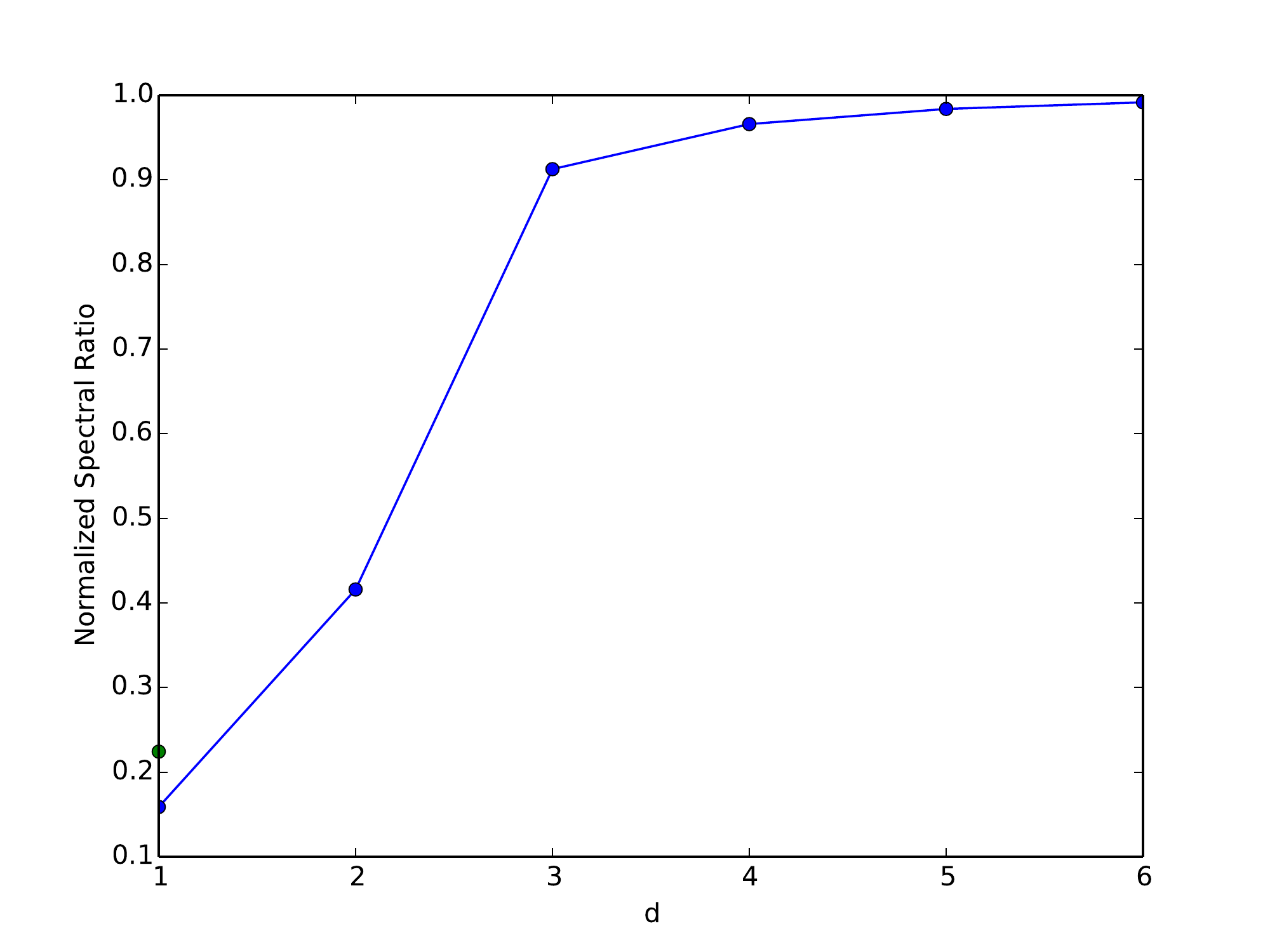}}
\subfigure[Jester]
{\includegraphics[width=6cm]{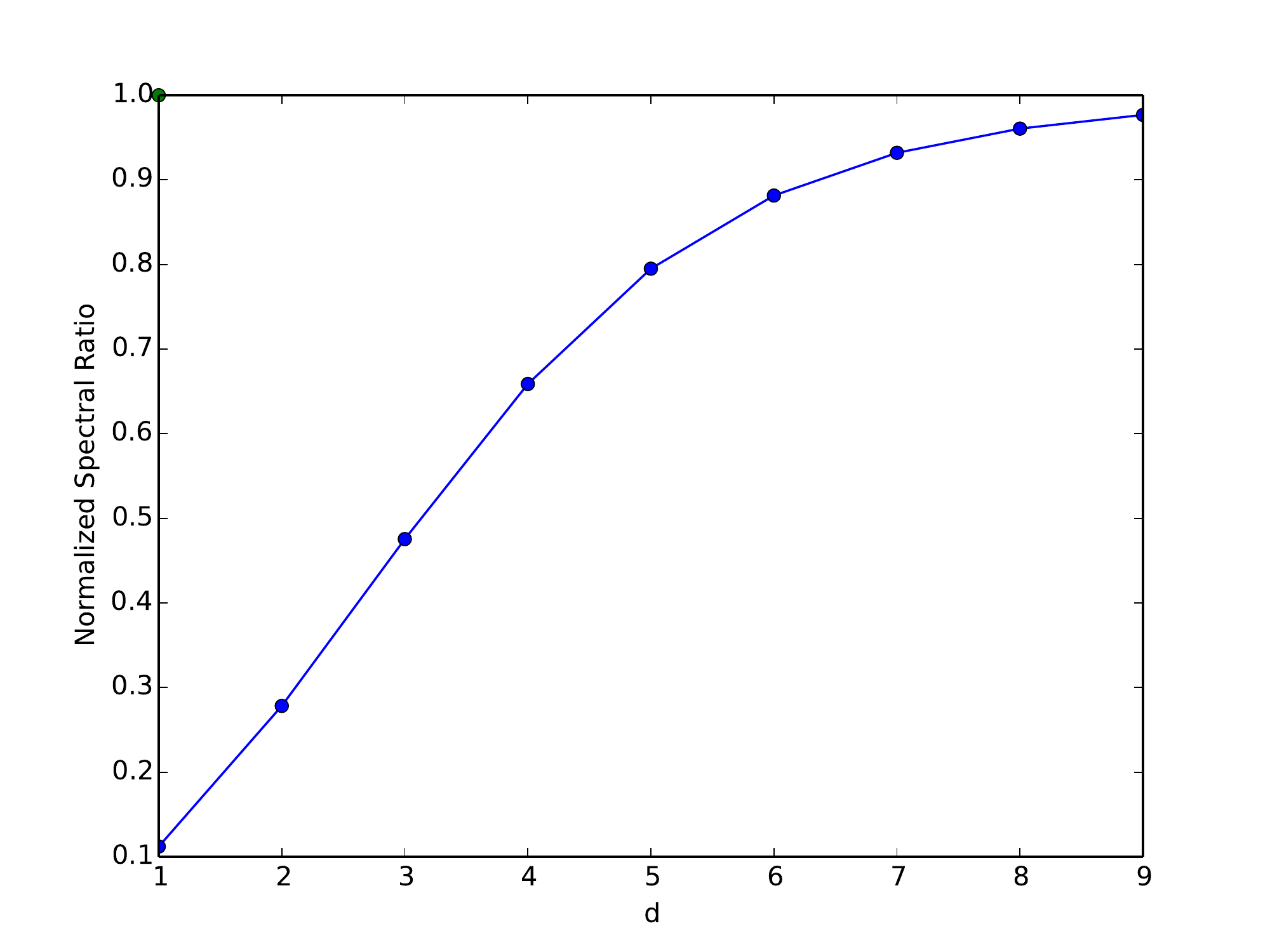}}
\subfigure[MovieLens]
{\includegraphics[width=6cm]{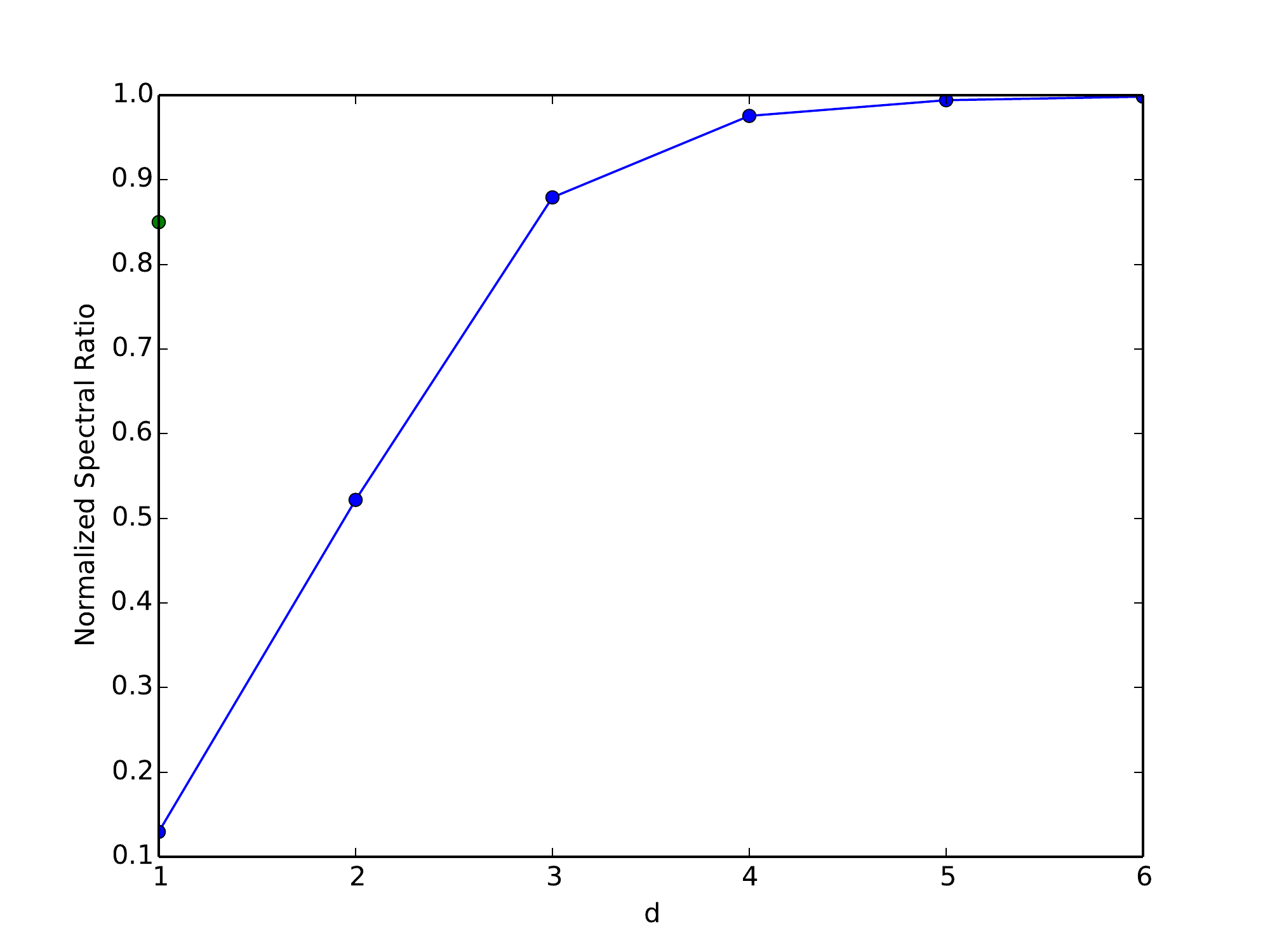}}
\subfigure[Netflix]
{\includegraphics[width=6cm]{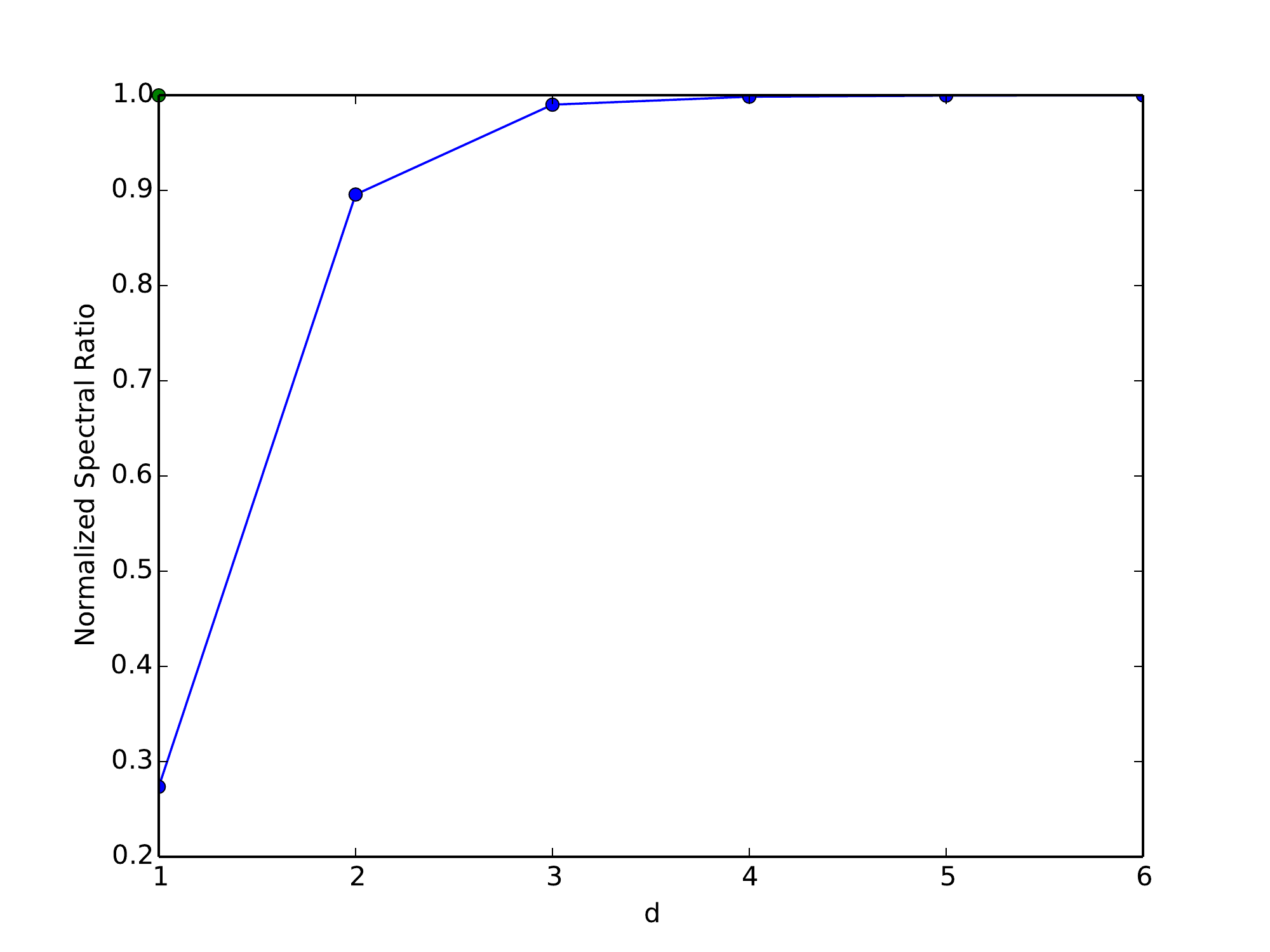}}
\caption{Normalized Spectral Ratio of the C-Kernel on different datasets. On the y-axis is indicated with a green dot the spectral ratio of the mDNF kernel.\label{fig:compc}}
\end{figure}

\begin{figure}[h!]
\centering
\subfigure[BookCrossing]
{\includegraphics[width=6cm]{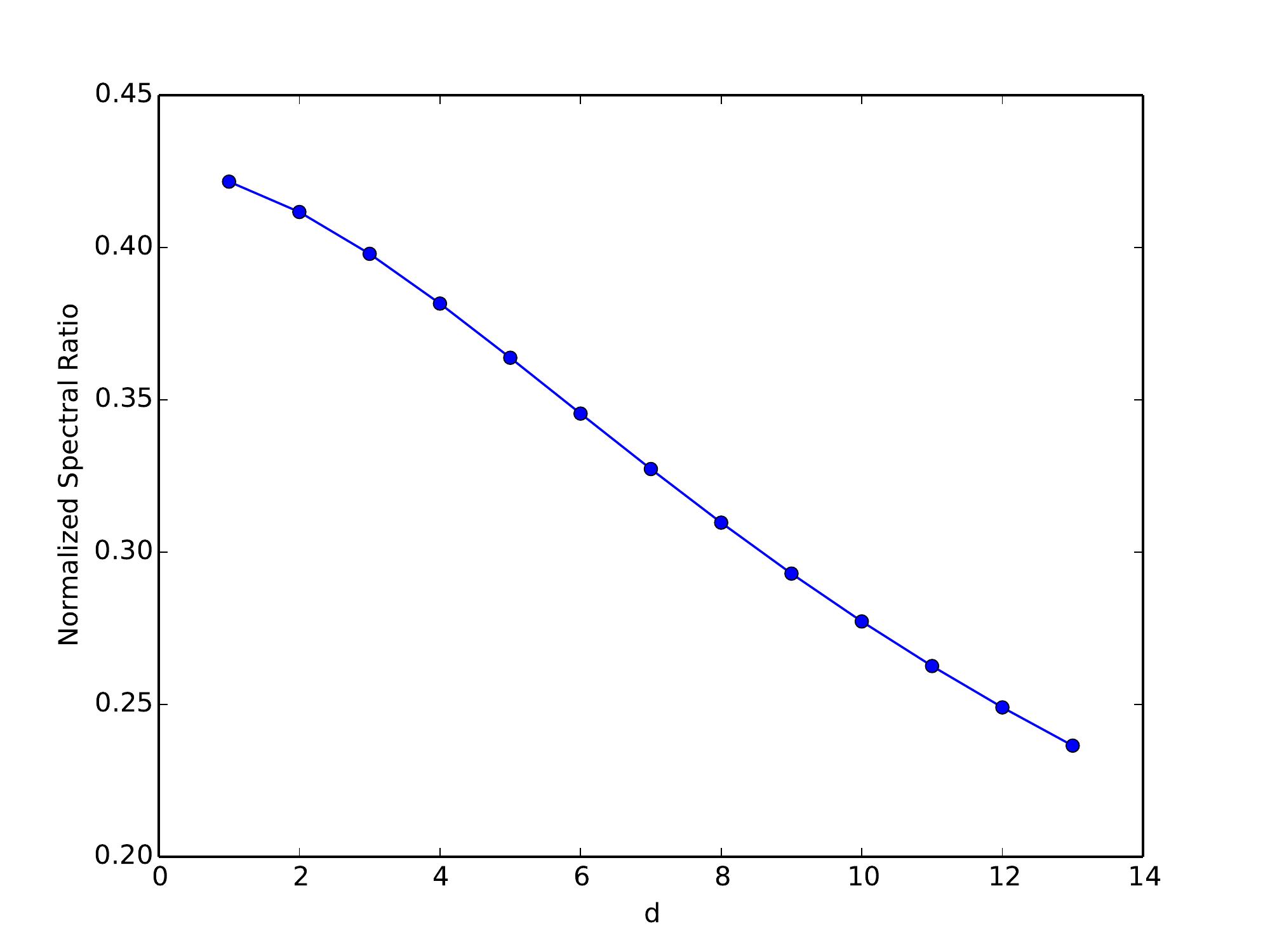}}
\subfigure[Ciao\label{fig:comp-ciao}]
{\includegraphics[width=6cm]{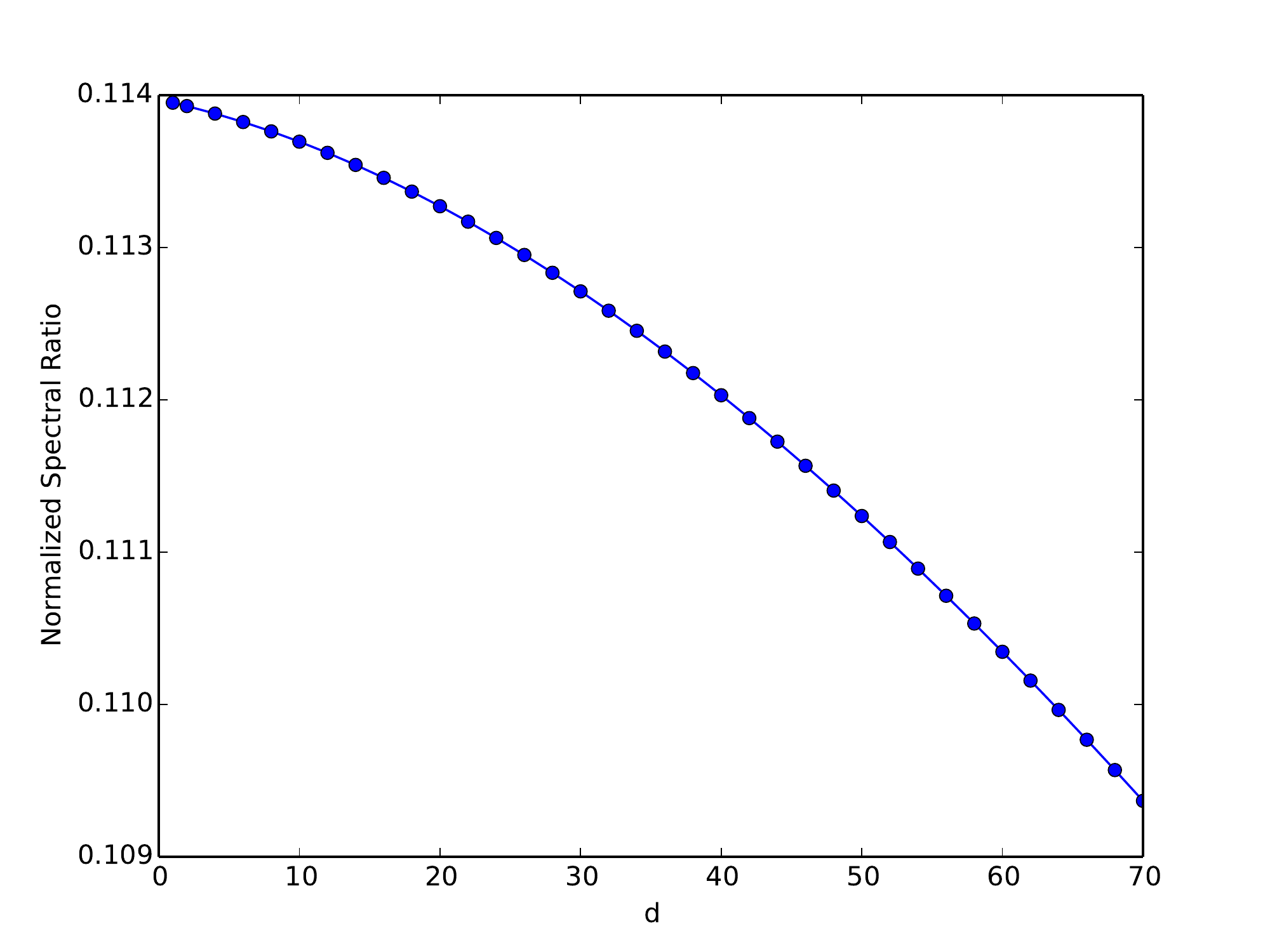}}
\subfigure[Film Trust]
{\includegraphics[width=6cm]{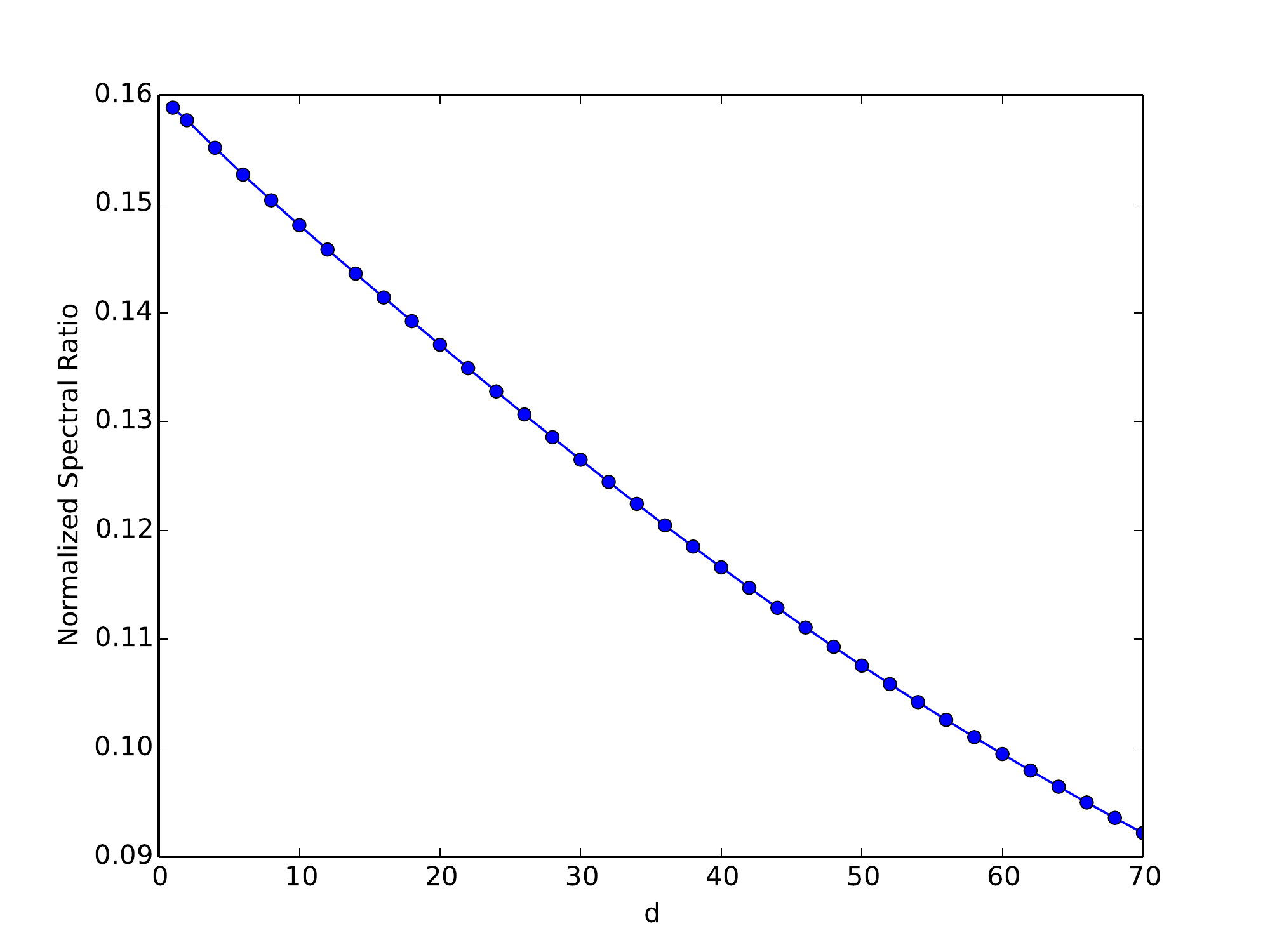}}
\subfigure[Jester]
{\includegraphics[width=6cm]{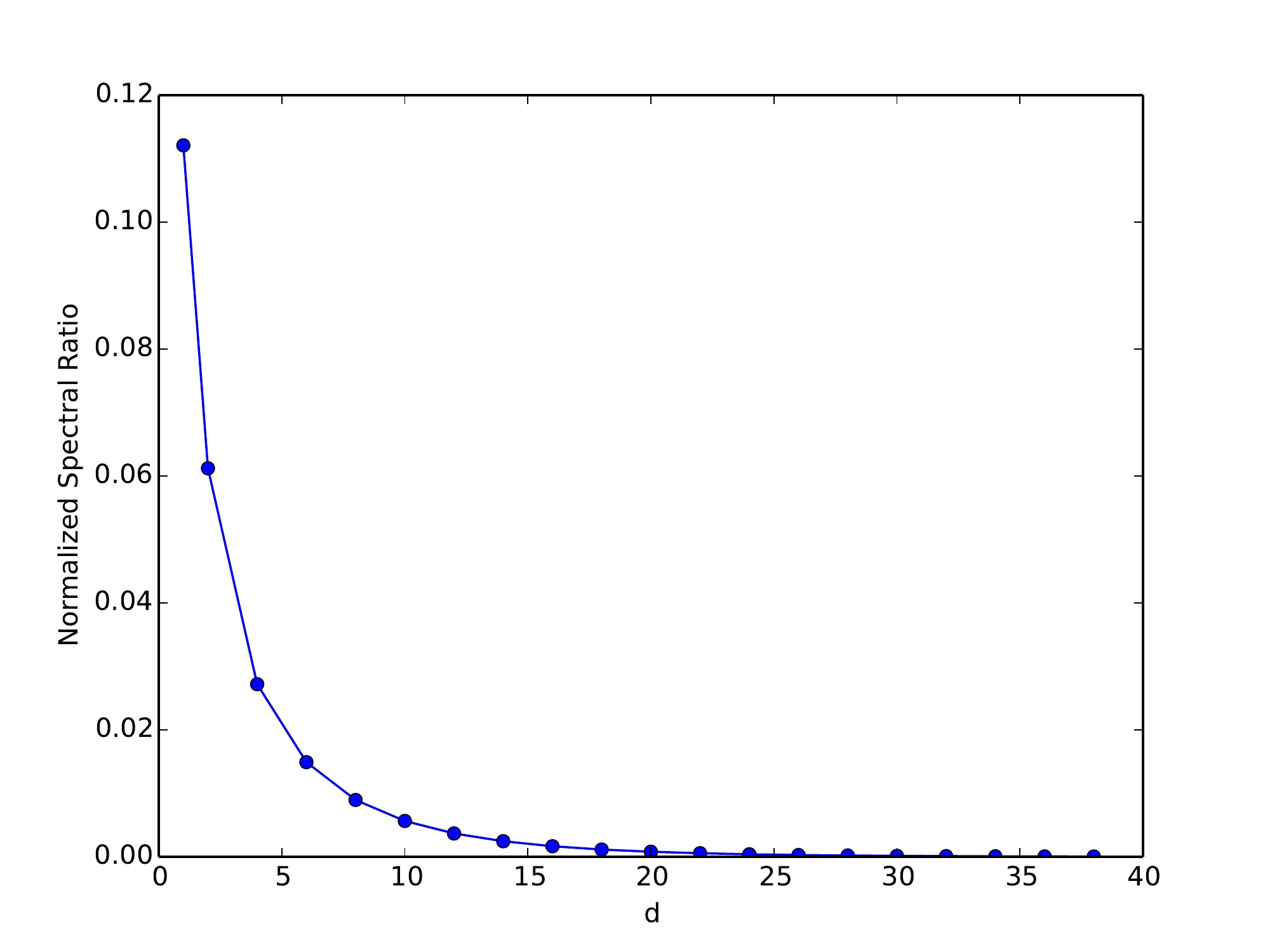}}
\subfigure[MovieLens]
{\includegraphics[width=6cm]{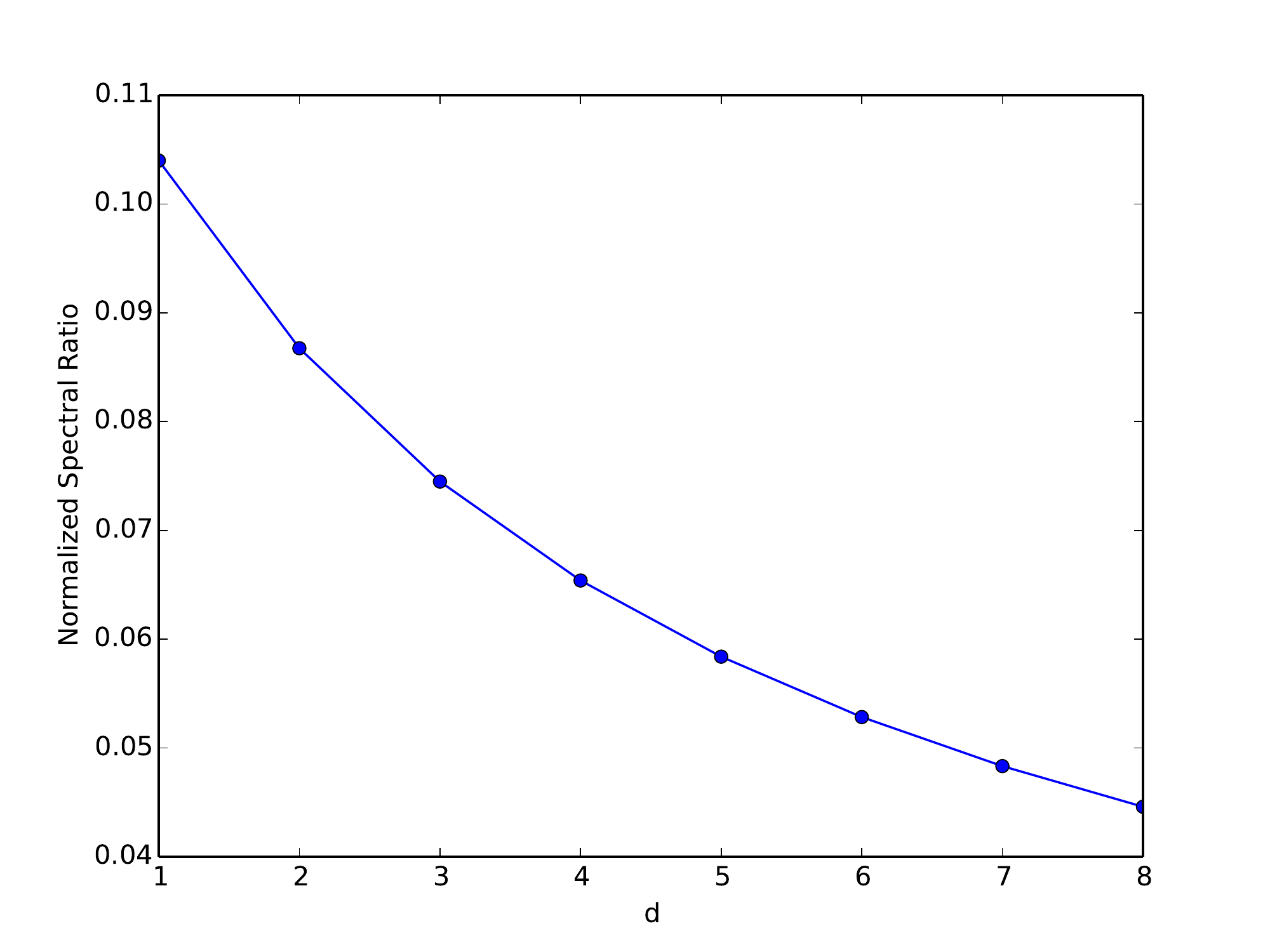}}
\subfigure[Netflix]
{\includegraphics[width=6cm]{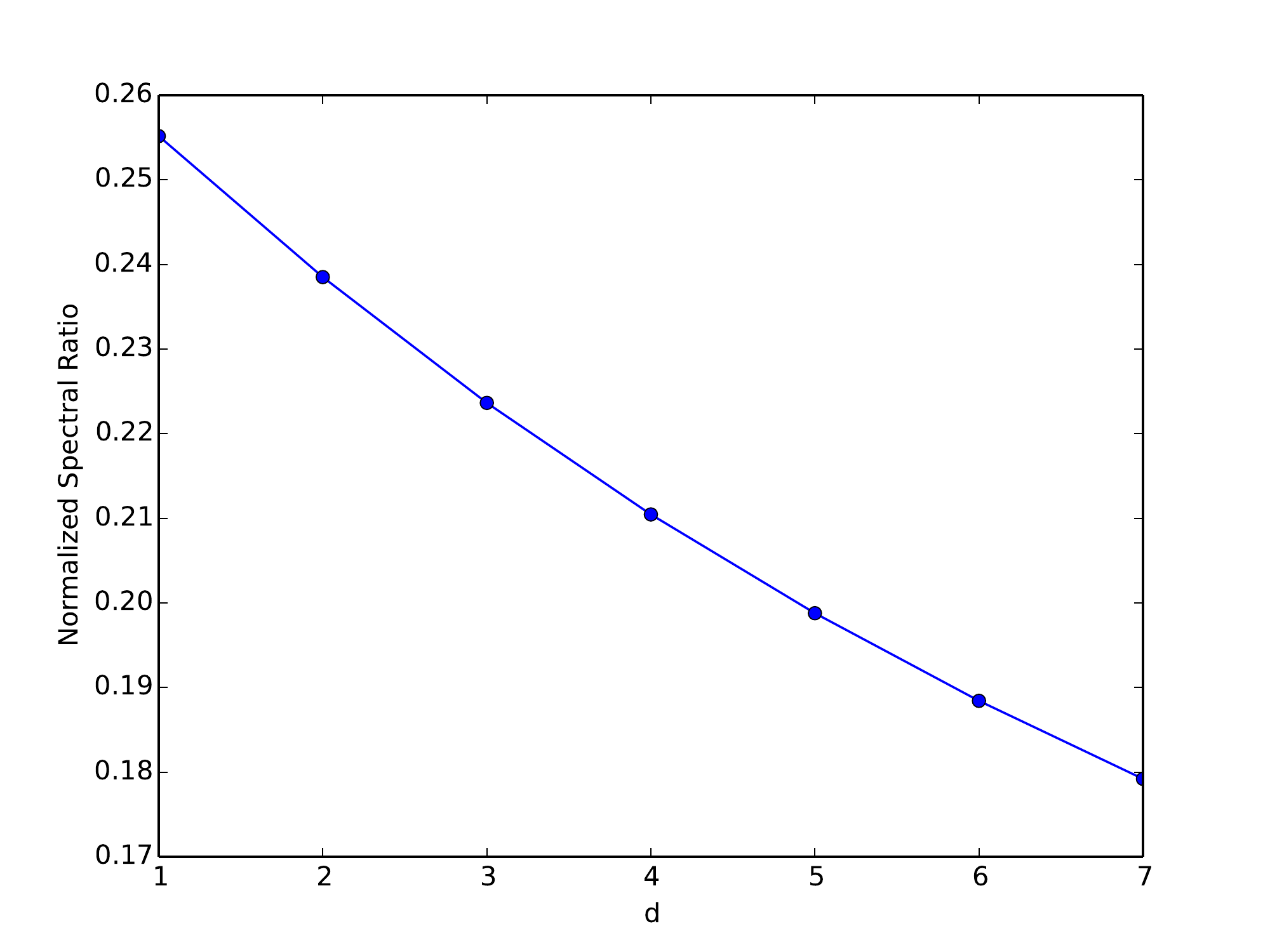}}
\caption{Normalized Spectral Ratio of the D-Kernel on different datasets.\label{fig:compd}}
\end{figure}

The reason for inclusion of the arity 1 (i.e., the linear kernel) in the plots is two-fold:
\begin{enumerate}
	\item it gives a visual idea of how fast/slow the expressiveness of the kernel increases/decreases with respect to the linear one;
	\item for each dataset, it offers a meeting point between the two curves.
\end{enumerate}

Besides supporting our previous results, the plots show that, generally, augmenting the arity, the decrease of expressiveness in D-Kernel is smoother than the increase of the  expressiveness in C-Kernel, which reaches very rapidly values near 1 (so the kernel matrix is similar to the identity). This justifies the fast decay of the C-Kernel' performance on all the datasets.

\subsection{Time complexity}\label{sec:time}
In Section \ref{sec:dkernel} we have demonstrated that the theoretical computational complexity of the D-Kernel is $\mathcal{O}(n+d)$. This complexity concerns only the computation of the kernel function between two binary vectors. 
However, in the calculation of the entire kernel matrix the number of examples (in this case the number of items) plays a huge role. Indeed, Figure \ref{fig:time} shows that all the kernel matrix computations, but the \verb#Ciao#, take similar time. \verb#Ciao# is the only dataset with a particular high number of items and consequently it has a very high number of entries in its kernel matrix. This justifies the difference in the computational time depicted in the figure.
\begin{figure}[h!]
\centering
{\includegraphics[width=5cm]{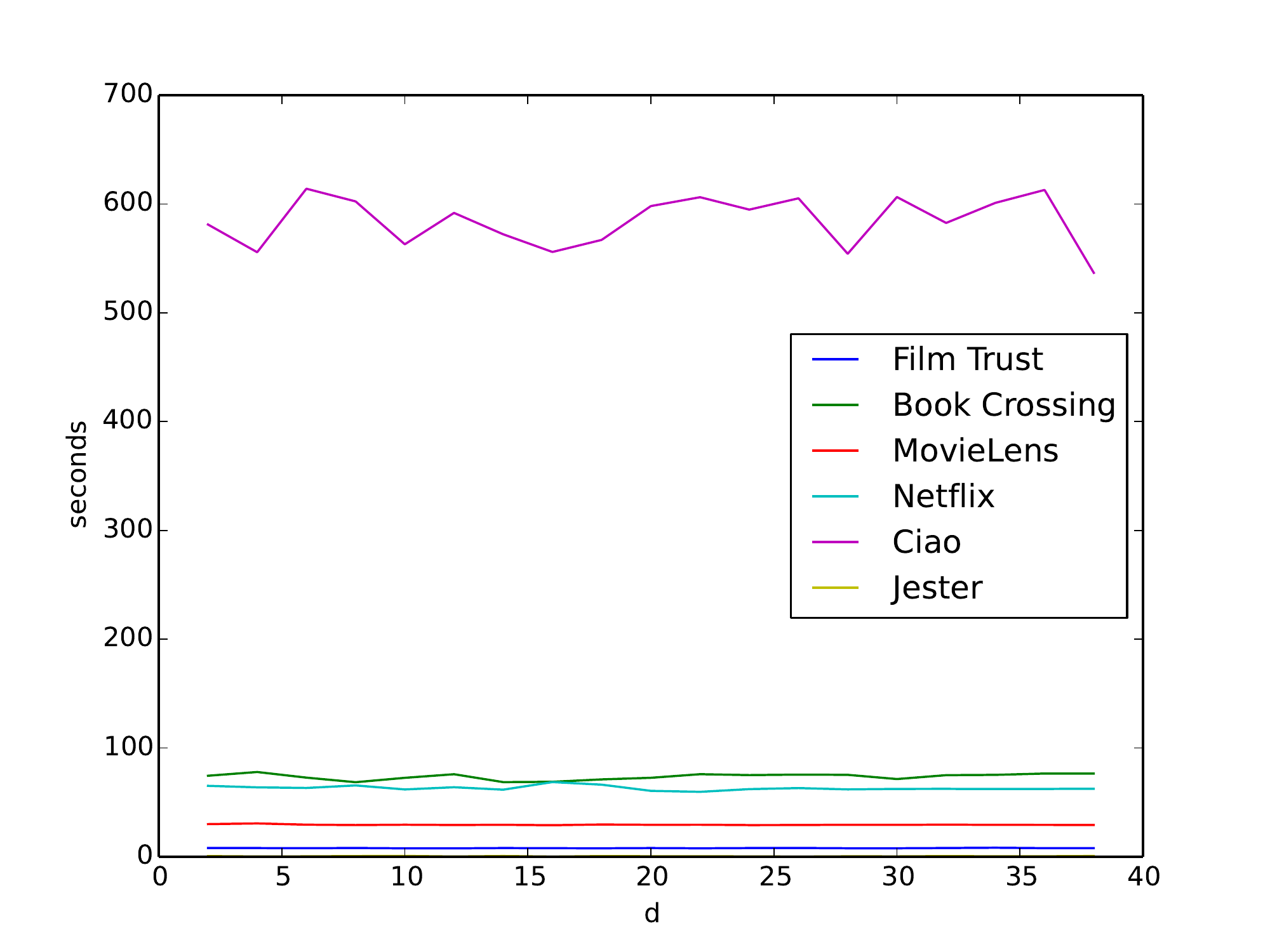}}
\caption{Time, in seconds, for calculating the D-Kernel on different datasets.\label{fig:time}}
\end{figure}

All these experiments have been made on a MacBook Pro late 2012 with 16GB of RAM and CPU Intel\textregistered~ Core i7 @ 2.70GHz.

\section{Conclusion and Future Work}\label{sec:concl}
In this work we have proposed a new boolean kernel, called D-Kernel, able to deal with the sparsity issue in a collaborative filtering context.
We leveraged on the observation made in our previous works \cite{Polato:2016}\cite{Polato:2017} to come up with the idea of creating a data representation less expressive than the linear one in order to mitigate the sparsity and the long tail issues that is common in CF datasets.
We presented a very efficient way for calculating the D-Kernel and we have also demonstrated its properties. The empirical analysis, over many CF datasets, show the effectiveness and the efficiency of the Disjunctive kernel against others, more expressive, boolean kernels. 
Our aim for the future is to enrich this boolean kernel family by introducing a new kernel which represents the bridge between the C-Kernel and the D-Kernel, for example by considering partially satisfied conjunctions of variables. 

\bibliographystyle{elsarticle-num}
\bibliography{library}

\begin{thebibliography}{10}
\expandafter\ifx\csname url\endcsname\relax
  \def\url#1{\texttt{#1}}\fi
\expandafter\ifx\csname urlprefix\endcsname\relax\def\urlprefix{URL }\fi
\expandafter\ifx\csname href\endcsname\relax
  \def\href#1#2{#2} \def\path#1{#1}\fi

\bibitem{Desrosiers:2011}
C.~Desrosiers, G.~Karypis, A Comprehensive Survey of Neighborhood-based
  Recommendation Methods, Springer US, Boston, MA, 2011, pp. 107--144.

\bibitem{Aiolli:2013}
F.~Aiolli, {Efficient top-N recommendation for very large scale binary rated
  datasets}, in: ACM Recommender Systems Conference, Hong Kong, China, 2013,
  pp. 273--280.

\bibitem{Polato-recsys:2016}
M.~Polato, F.~Aiolli, A preliminary study on a recommender system for the job
  recommendation challenge, in: Proceedings of the Recommender Systems
  Challenge, RecSys Challenge '16, ACM, New York, NY, USA, 2016, pp. 1:1--1:4.

\bibitem{Karypis:2011}
X.~Ning, G.~Karypis, {SLIM:} sparse linear methods for top-n recommender
  systems, in: 11th {IEEE} International Conference on Data Mining, {ICDM}
  2011, Vancouver, BC, Canada, December 11-14, 2011, 2011, pp. 497--506.

\bibitem{Koren:2008}
Y.~Hu, Y.~Koren, C.~Volinsky, Collaborative filtering for implicit feedback
  datasets, in: ICDM, 2008, pp. 263--272.

\bibitem{Aiolli:2014}
F.~Aiolli, {Convex AUC optimization for top-N recommendation with implicit
  feedback}, in: ACM Recommender Systems Conference, New York, USA, 2014, pp.
  293--296.

\bibitem{Polato:2016}
M.~Polato, F.~Aiolli, Kernel based collaborative filtering for very large scale
  top-n item recommendation, in: Proceedings of the European Symposium on
  Artificial Neural Networks, Computational Intelligence and Machine Learning
  (ESANN), 2016, pp. 11--16.

\bibitem{Sedhain:2016}
S.~Sedhain, A.~K. Menon, S.~Sanner, D.~Braziunas, On the effectiveness of
  linear models for one-class collaborative filtering, in: Proceedings of the
  Thirtieth AAAI Conference on Artificial Intelligence, AAAI'16, AAAI Press,
  2016, pp. 229--235.

\bibitem{Christakopoulou:2016}
E.~Christakopoulou, G.~Karypis, Local item-item models for top-n
  recommendation, in: Proceedings of the 10th ACM Conference on Recommender
  Systems, RecSys '16, ACM, New York, NY, USA, 2016, pp. 67--74.

\bibitem{Bresler:2014}
G.~Bresler, G.~H. Chen, D.~Shah, A latent source model for online collaborative
  filtering, in: Proceedings of the 27th International Conference on Neural
  Information Processing Systems, NIPS'14, MIT Press, Cambridge, MA, USA, 2014,
  pp. 3347--3355.

\bibitem{Polato:2017}
M.~Polato, F.~Aiolli, Exploiting sparsity to build efficient kernel based
  collaborative filtering for top-n item recommendation.

\bibitem{Zhou:2010}
J.~Zhou, T.~Luo, A novel approach to solve the sparsity problem in
  collaborative filtering, in: 2010 International Conference on Networking,
  Sensing and Control (ICNSC), 2010, pp. 165--170.

\bibitem{Grcar:2006}
M.~Gr{\v{c}}ar, D.~Mladeni{\v{c}}, B.~Fortuna, M.~Grobelnik, Data Sparsity
  Issues in the Collaborative Filtering Framework, Springer Berlin Heidelberg,
  Berlin, Heidelberg, 2006, pp. 58--76.

\bibitem{Donini:2016}
M.~Donini, F.~Aiolli, Learning deep kernels in the space of dot product
  polynomials, Machine Learning (2016) 1--25.

\bibitem{Aiolli:2005}
F.~Aiolli, A preference model for structured supervised learning tasks, IEEE
  International Conference on Data Mining (2005) 557--560.

\bibitem{Aiolli:2008}
F.~Aiolli, G.~Da~San~Martino, A.~Sperduti, A Kernel Method for the Optimization
  of the Margin Distribution, Springer Berlin Heidelberg, Berlin, Heidelberg,
  2008, pp. 305--314.

\bibitem{Zhang:2003}
Y.~Zhang, Z.~Li, M.~Kang, J.~Yan, {Improving the classification performance of
  boolean kernels by applying Occam's razor}.

\bibitem{Zhang:2005}
Y.~Zhang, Z.~Li, K.~Cui, DRC-BK: Mining Classification Rules by Using Boolean
  Kernels, Springer Berlin Heidelberg, Berlin, Heidelberg, 2005, pp. 214--222.

\bibitem{Shawe-Taylor:2004}
J.~Shawe-Taylor, N.~Cristianini, Kernel Methods for Pattern Analysis, Cambridge
  University Press, New York, NY, USA, 2004.

\bibitem{Ralaivola:2005}
L.~Ralaivola, S.~J. Swamidass, H.~Saigo, P.~Baldi, Graph kernels for chemical
  informatics, Neural Networks 18~(8) (2005) 1093 -- 1110, neural Networks and
  Kernel Methods for Structured Domains.

\bibitem{Sadohara:2001}
K.~Sadohara, Learning of Boolean Functions Using Support Vector Machines,
  Springer Berlin Heidelberg, Berlin, Heidelberg, 2001, pp. 106--118.

\bibitem{KhardonRS:2005}
R.~Khardon, D.~Roth, R.~A. Servedio, Efficiency versus convergence of boolean
  kernels for on-line learning algorithms, J. Artif. Intell. Res. {(JAIR)} 24
  (2005) 341--356.

\bibitem{Sadohara:2002}
K.~Sadohara, On a capacity control using boolean kernels for the learning of
  boolean functions, in: Proceedings of the 2002 IEEE International Conference
  on Data Mining, ICDM '02, IEEE Computer Society, Washington, DC, USA, 2002,
  pp. 410--.

\bibitem{KowalczykSW:2001}
A.~Kowalczyk, A.~J. Smola, R.~C. Williamson, Kernel machines and boolean
  functions, in: Advances in Neural Information Processing Systems 14 [Neural
  Information Processing Systems: Natural and Synthetic, {NIPS} 2001, December
  3-8, 2001, Vancouver, British Columbia, Canada], 2001, pp. 439--446.

\bibitem{Khardon:2003}
R.~Khardon, R.~A. Servedio, Maximum Margin Algorithms with Boolean Kernels,
  Springer Berlin Heidelberg, Berlin, Heidelberg, 2003, pp. 87--101.

\bibitem{Guo:2014}
G.~Guo, J.~Zhang, D.~Thalmann, N.~Yorke-Smith, Etaf: An extended trust
  antecedents framework for trust prediction, in: Proceedings of the 2014
  International Conference on Advances in Social Networks Analysis and Mining
  (ASONAM), 2014, pp. 540--547.

\bibitem{Guo:2013}
G.~Guo, J.~Zhang, N.~Yorke-Smith, A novel bayesian similarity measure for
  recommender systems, in: Proceedings of the 23rd International Joint
  Conference on Artificial Intelligence (IJCAI), 2013, pp. 2619--2625.

\bibitem{Goldberg:2001}
K.~Goldberg, T.~Roeder, D.~Gupta, C.~Perkins, Eigentaste: A constant time
  collaborative filtering algorithm, Information Retrieval 4~(2) (2001)
  133--151.

\end{thebibliography}

\end{document}